\documentclass[journal,10pt,letterpaper]{IEEEtran}

\usepackage{epstopdf}
\usepackage{graphicx}
\usepackage{stmaryrd}
\usepackage{amsmath,amsthm,amssymb}
\usepackage{multirow,url}
\usepackage{color,cite}
\usepackage{algorithm}
\usepackage{algorithmicx}
\usepackage{algpseudocode}
\usepackage{setspace}
\usepackage{rotating}
\usepackage{bm}
\usepackage{epsfig}
\usepackage{pdfsync}

\newtheorem{theorem}{Theorem}
\newtheorem{definition}{Definition}

\newtheorem{lemma}{Lemma}
\newtheorem{corollary}{Corollary}

\ifodd 0
\newcommand{\rev}[1]{{\color{blue}#1}} 
\newcommand{\com}[1]{\textbf{\color{red} (COMMENT: #1) }} 
\newcommand{\comg}[1]{\textbf{\color{green} (COMMENT: #1)}}
\newcommand{\response}[1]{\textbf{\color{green} (RESPONSE: #1)}} 
\else
\newcommand{\rev}[1]{#1}

\newcommand{\com}[1]{}
\newcommand{\comg}[1]{}
\newcommand{\response}[1]{}
\fi

\begin{document}

\title{Efficient Multi-User Computation Offloading for Mobile-Edge Cloud Computing}

\author{Xu Chen, \emph{Member, IEEE},  Lei Jiao, \emph{Member, IEEE}, Wenzhong Li, \emph{Member, IEEE}, and Xiaoming Fu, \emph{Senior Member, IEEE} 
%
%
%
%
%
}

\maketitle
\pagestyle{empty}
\thispagestyle{empty}

\allowdisplaybreaks

\begin{abstract}
Mobile-edge cloud computing is a new paradigm to provide cloud computing capabilities at the edge of pervasive radio access networks in close proximity to mobile users. \rev{In this paper, we first study the multi-user computation offloading problem for mobile-edge cloud computing in a multi-channel wireless interference environment.} We show that it is NP-hard to compute a centralized optimal solution, and hence adopt a game theoretic approach for achieving efficient computation offloading in a distributed manner. We formulate the distributed computation offloading decision making problem among mobile device users as a multi-user computation offloading game. We analyze the structural property of the game and show that the game admits a Nash equilibrium and possesses the finite improvement property. We then design a distributed computation offloading algorithm that can achieve a Nash equilibrium, derive the upper bound of the convergence time, and quantify its efficiency ratio over the centralized optimal solutions in terms of two important performance metrics. \rev{We further extend our study to the scenario of multi-user computation offloading in the multi-channel wireless contention environment.} Numerical results corroborate that the proposed algorithm can achieve superior computation offloading performance and scale well as the user size increases.
\end{abstract}


\section{Introduction}

As smartphones are gaining enormous popularity, more and more new
mobile applications such as face recognition, natural language processing,
interactive gaming, and augmented reality are emerging and attract
great attention \cite{kumar2010cloud,soyata2012cloud,cohen2008embedded}. This kind of mobile applications
are typically resource-hungry, demanding intensive computation and
high energy consumption. Due to the physical size constraint, however,
mobile devices are in general resource-constrained, having limited
computation resources and battery life. The tension between
resource-hungry applications and resource-constrained mobile devices
hence poses a significant challenge for the future mobile platform development \cite{cuervo2010maui}.


Mobile cloud computing is envisioned as a promising
approach to address such a challenge. By offloading the
computation via wireless access to the resource-rich cloud
infrastructure, mobile cloud computing can augment the capabilities
of mobile devices for resource-hungry applications.
One possible approach is to offload the computation to the remote
public clouds such as Amazon EC2 and Windows Azure.
However, an evident weakness of public cloud based mobile
cloud computing is that mobile users may experience long
latency for data exchange with the public cloud through the
wide area network. Long latency would hurt the interactive
response, since humans are acutely sensitive to delay and jitter.
Moreover, it is very difficult to reduce the latency in the wide
area network. To overcome this limitation, the cloudlet based mobile cloud computing was proposed as a promising solution \cite{satyanarayanan2009case}. Rather than relying on a remote cloud, the cloudlet based
mobile cloud computing leverages the physical proximity to reduce delay by offloading the computation to the nearby computing sever/cluster via one-hop WiFi wireless access. However, there are two major disadvantages for the cloudlet based mobile cloud computing: 1) due to limited coverage of WiFi networks (typically available for indoor environments), cloudlet based mobile cloud computing can not guarantee  ubiquitous service provision everywhere; 2) due to space constraint, cloudlet based mobile cloud computing usually utilizes a computing sever/cluster with small/medium computation resources, which may not satisfy QoS of a large number of users.

\begin{figure}[t]
\centering
\includegraphics[scale=0.45]{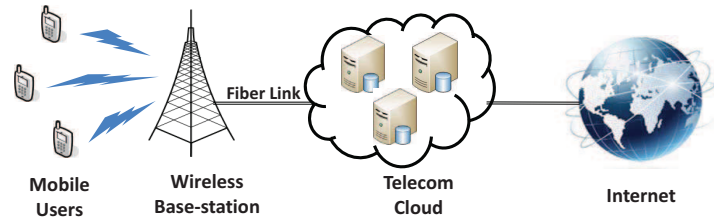}
\caption{\label{fig:An-illustration-of}An illustration of mobile-edge cloud computing }
\end{figure}

To address these challenges and complement cloudlet based mobile cloud computing, a novel mobile cloud computing paradigm, called mobile-edge cloud computing, has been proposed \cite{MEC2014,drolia2013case,TelecomCloud2012,barbarossa2013joint}. As illustrated in Figure \ref{fig:An-illustration-of}, mobile-edge cloud computing can provide cloud-computing capabilities at the edge of pervasive radio access networks in close proximity to mobile users.  In this case, the need for fast interactive response can be met by fast  and  low-latency connection (e.g., via fiber transmission) to large-scale resource-rich cloud computing infrastructures (called telecom cloud) deployed by telecom operators (e.g., AT\&T and T-Mobile) within the network edge and backhaul/core networks. By endowing ubiquitous radio access networks (e.g., 3G/4G macro-cell and small-cell base-stations) with powerful computing capabilities, mobile-edge cloud computing is envisioned to provide pervasive and agile computation augmenting services for mobile device users at anytime and anywhere \cite{MEC2014,drolia2013case,TelecomCloud2012,barbarossa2013joint}.


In this paper, we study the issue of designing efficient computation offloading mechanism for mobile-edge cloud computing. One critical factor of affecting the computation offloading performance is the wireless access efficiency \cite{barbera2013offload}. Given the fact that base-stations in most wireless networks are operating in multi-channel setting, a key challenge is how to achieve efficient wireless access coordination among multiple mobile device users for computation offloading. If too many mobile device users choose the same wireless channel to offload the computation to the cloud
simultaneously, they may cause severe interference to each other, which would reduce the data rates for computation offloading. This hence can lead to low energy efficiency and long data transmission time. In this case, it would not be beneficial for the mobile device users to offload computation to the cloud. \rev{To achieve efficient computation offloading for mobile-edge cloud computing, we hence need to carefully tackle two key challenges: 1) how should a mobile user choose between the local computing (on its own device) and the cloud computing (via computation offloading)? 2) if a user chooses the cloud computing, how can the user choose a proper channel in order to achieve high wireless access efficiency for computation offloading?}

We adopt a game theoretic approach to address these challenges. Game theory is a powerful tool for designing distributed
mechanisms, such that the mobile device users in the system can locally
make decisions based on strategic interactions and achieve a mutually satisfactory computation offloading solution. This can help to ease the heavy burden of complex centralized management (e.g., massive information collection from mobile device users) by the telecom cloud operator. Moreover, as different mobile devices are usually owned by different individuals and they may pursue different interests, game theory provides a useful framework to analyze
the interactions among multiple mobile device users who act in their own interests and devise incentive compatible computation offloading mechanisms such that no mobile user has the incentive to deviate unilaterally.

Specifically, we model the computation offloading decision making problem among multiple mobile device users for mobile-edge cloud computing in a multi-channel wireless environment as a multi-user computation offloading game. We then propose a
distributed computation offloading algorithm that can achieve the Nash equilibrium of the game. The main results and contributions of
this paper are as follows:
\begin{itemize}
\item \emph{Multi-User Computation Offloading Game Formulation}: \rev{We first show that it is NP-hard to find the centralized optimal multi-user computation offloading solutions in a multi-channel wireless interference environment. We hence consider the distributed alternative and formulate
the distributed computation offloading decision making problem among the mobile device users as a multi-user computation offloading game, which takes into account both communication and computation aspects of mobile-edge cloud computing. We also extend our study to the scenario of multi-user computation offloading in the multi-channel wireless contention environment.}
\item \emph{Analysis of Computation Offloading Game Properties}: We then study the structural property of the multi-user computation offloading game and show that the game is a potential game by carefully constructing a potential function. According to the property of potential game, we show that the  multi-user computation offloading game admits the finite improvement property and always possesses a Nash equilibrium.
\item \emph{Distributed Computation Offloading Algorithm Design}: We next devise a distributed computation offloading algorithm that achieves a Nash equilibrium of the multi-user computation offloading game and derive the upper bound of the convergence time under mild conditions. We further quantify the efficiency ratio of the Nash equilibrium solution by the algorithm over the centralized optimal solutions in terms of two important metrics of the number of beneficial cloud computing users and the system-wide computation overhead. Numerical results demonstrate that the proposed algorithm can achieve efficient computation offloading performance and scale well as the user size increases.
\end{itemize}

The rest of the paper is organized as follows. We first  present the system model
in Section \ref{sec:System-Model}. We then propose the multi-user
computation offloading game and develop the distributed computation
offloading algorithm in Sections \ref{sec:Decentralized-Computation-Offloa}
and \ref{sec:Decentralized-Computation-Offloa-1}, respectively. We next analyze the performance of the algorithm and present the numerical results in Sections \ref{performance} and \ref{sec:Numerical-Results}, respectively. We further extend our study to the case under the wireless contention model in Section \ref{Extension}, discuss the related
work in Section \ref{sec:Related-Work}, and finally conclude in Section \ref{sec:Conclusion}.

\section{System Model}\label{sec:System-Model}
We first introduce the system model. We consider a set of $\mathcal{N}=\{1,2,...,N\}$ collocated mobile device users, where each user has a computationally intensive task to be completed. There exists a wireless base-station $s$ and through which the mobile device users can offload the computation to the cloud in proximity deployed by the telecom operator. Similar to many previous studies in mobile cloud computing (e.g., \cite{barbera2013offload,rudenko1998saving,huerta2008adaptable,xian2007adaptive,huang2012dynamic,wen2012energy,wu2013making,Chen2014DCO,barbarossa2013joint}) and mobile networking (e.g., \cite{wu2002multi} and \cite{iosifidisiterative2013}), to enable tractable analysis and get useful insights, we consider a quasi-static scenario where the set of mobile device users $\mathcal{N}$ remains unchanged  during a computation offloading period (e.g., several hundred milliseconds), while may change across different periods\footnote{The general case that mobile users may depart and leave dynamically within a computation offloading period will be considered in a future work.}.  Since both the communication and computation aspects play a key role in mobile-edge cloud computing, we next introduce the communication and computation models in details.

\subsection{Communication Model}\label{sub:Communication-Model}

We first introduce the communication model for wireless access in mobile-edge cloud computing. Here the wireless base-station $s$ can be a 3G/4G macro-cell or small-cell base-station \cite{lopez2013distributed} that manages the uplink/downlink communications of mobile device users. There are $M$ wireless channels and the set of channels is denoted as $\mathcal{M}=\{1,2,...,M\}$. Furthermore, we denote $a_{n}\in\{0\}\cup\mathcal{M}$ as the computation offloading decision of mobile device user $n$. Specifically, we have $a_{n}>0$ if user $n$ chooses to offload the computation to the cloud via a wireless channel $a_{n}$; we have $a_{n}=0$ if user $n$ decides to compute its task locally on its own mobile device.
Given the decision profile $\boldsymbol{a}=(a_{1},a_{2},...,a_{N})$
of all the mobile device users, we can compute the uplink data rate of a mobile device user $n$ that chooses to offload the computation to the cloud via a wireless channel $a_{n}>0$ as \cite{rappaport1996wireless}
\begin{align}
r_{n}(\boldsymbol{a})=w\log_{2}\left(1+\frac{q_{n}g_{n,s}}{\varpi_{0}+\sum_{i\in\mathcal{N}\backslash\{n\}:a_{i}=a_{n}}q_{i}g_{i,s}}\right).\label{eq:R1}
\end{align}
Here $w$ is the channel bandwidth and $q_{n}$ is user $n$'s transmission
power which is determined by the wireless base-station according to some power
control algorithms such as \cite{xiao2003utility} and \cite{chiang2008power}\footnote{To be compatible with existing wireless systems, in this paper we consider that the power is determined to satisfy the requirements of wireless transmission (e.g., the specified SINR threshold). For the future work, we will study the joint power control and offloading decision making problem to optimize the performance of computation offloading. This joint problem would be very challenging to solve since the offloading decision making problem alone is NP-hard as we show later.}. Further, $g_{n,s}$ denotes
the channel gain between the mobile device user $n$ and the base-station $s$,
and $\varpi_{0}$ denotes the background noise power. \rev{Note that here we focus on exploring the computation offloading problem under the wireless interference model, which can well capture user's time average aggregate throughput in the cellular communication scenario in which some physical layer channel access
scheme (e.g., CDMA) is adopted to allow multiple users to share the same spectrum resource simultaneously and efficiently. In Section \ref{Extension}, we will also extend our study to the wireless contention model in which some media access control protocol such as CSMA is adopted in WiFi-alike networks.}

From the communication model in (\ref{eq:R1}), we see that if too many mobile device users  choose to offload the computation via the same wireless access channel simultaneously during a computation offloading period, they may incur severe interference, leading to low data rates. As we discuss latter, this would negatively affect the performance of mobile-edge cloud computing.

\subsection{Computation Model}

We then introduce the computation model. We consider that each mobile device
user $n$ has a computation task $\mathcal{J}_{n}\triangleq(b_{n},d_{n})$
that can be computed  either locally on the mobile device or remotely
on the telecom cloud via computation offloading. Here $b_{n}$ denotes the
size of computation input data (e.g., the program codes and input parameters) involved in the computation task $\mathcal{J}_{n}$
and $d_{n}$ denotes the total number of CPU cycles required to accomplish the computation task $\mathcal{J}_{n}$. A mobile device user $n$ can apply the methods (e.g., call graph analysis) in \cite{cuervo2010maui,yang2013framework} to obtain the information of $b_{n}$ and $d_{n}$.   We
next discuss the computation overhead in terms of both energy consumption and processing time for both local and cloud computing approaches.

\subsubsection{Local Computing}

For the local computing approach, a mobile device user $n$ executes its
computation task $\mathcal{J}_{n}$ locally on the mobile device. Let
$f^{m}_{n}$ be the computation capability (i.e., CPU cycles per
second) of mobile device user $n$. Here we allow that different mobile devices may have different computation capabilities. The computation execution time of
the task $\mathcal{J}_{n}$ by local computing is then given as
\begin{equation}
t_{n}^{m}=\frac{d_{n}}{f^{m}_{n}}.\label{eq:l1}
\end{equation}
For the computational energy, we have that
\begin{equation}
e_{n}^{m}=\gamma_{n}d_{n},\label{eq:l2}
\end{equation}
where $\gamma_{n}$ is the coefficient denoting the consumed energy per CPU cycle, which can be obtained by the measurement method in \cite{wen2012energy}.

According to (\ref{eq:l1}) and (\ref{eq:l2}), we can then compute
the overhead of the local computing approach in terms of computational
time and energy as
\begin{equation}
K_{n}^{m}=\lambda_{n}^{t}t_{n}^{m}+\lambda_{n}^{e}e_{n}^{m},\label{eq:l3}
\end{equation}
where $\lambda_{n}^{t},\lambda_{n}^{e}\in \{0,1\}$ denote the weighting parameters of
computational time and energy for mobile device user $n$'s decision making,
respectively. When a user is at a low battery state and cares about the energy consumption, the user can set $\lambda_{n}^{e}=1$ and $\lambda_{n}^{t}=0$ in the decision making. When a user is running some application that is sensitive to the delay (e.g., video streaming) and hence concerns about the processing time, then the user can set $\lambda_{n}^{e}=0$ and $\lambda_{n}^{t}=1$ in the decision making. To provide rich modeling flexibility, our model can also apply to the generalized case where $\lambda_{n}^{t},\lambda_{n}^{e}\in [0,1]$ such that  a user can take both computational time and energy into the decision making at the same time. In practice the proper weights that capture a user's valuations on computational energy and time can be determined by applying the multi-attribute utility approach in the multiple criteria decision making theory \cite{wallenius2008multiple}.

\subsubsection{Cloud Computing}

For the cloud computing approach, a mobile device user $n$ will offload
its computation task $\mathcal{J}_{n}$ to the cloud in proximity deployed by telecom operator via wireless access and the cloud
will execute the computation task on behalf of the mobile device user.

For the computation offloading, a mobile device user $n$ would incur the extra
overhead in terms of time and energy for transmitting the computation input data
to the cloud via wireless access. According to the communication
model in Section \ref{sub:Communication-Model}, we can compute the
transmission time and energy of mobile device user $n$ for offloading the input
data of size $b_{n}$ as, respectively,
\begin{equation}
t_{n,off}^{c}(\boldsymbol{a})=\frac{b_{n}}{r_{n}(\boldsymbol{a})},\label{eq:c2}
\end{equation}
and
\begin{equation}
e_{n}^{c}(\boldsymbol{a})=\frac{q_{n}b_{n}}{r_{n}(\boldsymbol{a})}+L_{n},\label{eq:c3}
\end{equation}
where $L_{n}$ is the tail energy due to that the mobile device will continue to hold the channel for a while even after the data transmission. Such a tail phenomenon is commonly observed in 3G/4G networks \cite{hu2014quality}. After the offloading, the cloud will execute the computation
task $\mathcal{J}_{n}$. We denote $f^{c}_{n}$ as the computation capability
(i.e., CPU cycles per second) assigned to user $n$ by the cloud. Similar to the mobile data usage service, the cloud computing capability $f^{c}_{n}$ is determined according to the cloud computing service contract subscribed by the mobile user $n$ from the telecom operator. Due to the fact many telecom operators (e.g., AT\&T and T-Mobile) are capable for large-scale cloud computing infrastructure investment, we consider that the cloud computing resource requirements of all users can be satisfied. The case that a small/medium telecom operator has limited cloud computing resource provision will be considered in a future work.   Then the execution time
of the task $\mathcal{J}_{n}$ of mobile device user $n$ on the cloud can
be then given as
\begin{equation}
t_{n,exe}^{c}=\frac{d_{n}}{f^{c}_{n}}.\label{eq:c1}
\end{equation}

According to (\ref{eq:c2}), (\ref{eq:c3}), and (\ref{eq:c1}), we
can compute the overhead of the cloud computing approach in terms
of processing time and energy as
\begin{equation}
K_{n}^{c}(\boldsymbol{a})=\lambda_{n}^{t}\left(t_{n,off}^{c}(\boldsymbol{a})+t_{n,exe}^{c}\right)+\lambda_{n}^{e}e_{n}^{c}(\boldsymbol{a}).\label{eq:c4}
\end{equation}

Similar to many studies such as \cite{huang2012dynamic,rudenko1998saving,huerta2008adaptable,xian2007adaptive}, we neglect
the time overhead for the cloud to send the computation outcome back
to the mobile device user, due to the fact that for many applications (e.g., face recognition), the size of the computation outcome in general is much smaller than the size of computation input data, which includes the mobile system settings, program codes and input parameters. \rev{Also, due to the fact that wireless spectrum is the most constrained resource, and higher-layer network resources are much richer and the higher-layer management can be done quickly and efficiently via high-speed wired connection and high-performance computing using powerful servers at the base-station, the wireless access efficiency at the physical layer is the bottleneck for computation offloading via wireless transmission \cite{barbera2013offload}. Similar to existing studies for mobile cloud computing \cite{Chen2014DCO,yang2013framework,barbarossa2013joint}, we hence account for the most critical factor (i.e., wireless access at the physical layer) only\footnote{We can account for the high-layer factors by simply adding a processing latency term (which is typically much smaller than the wireless access) into user’s time overhead function and this will not affect the analysis of the problem.}.}

Based on the system model above,  in the following sections we will develop a game theoretic approach for devising efficient multi-user computation offloading policy for the mobile-edge cloud computing.

\section{Multi-User Computation Offloading Game}\label{sec:Decentralized-Computation-Offloa}

In this section, we consider the issue of achieving efficient multi-user computation offloading for the mobile-edge cloud computing.

According to the communication and computation models in Section \ref{sec:System-Model}, we see
that the computation offloading decisions $\boldsymbol{a}$ among the mobile device users
are coupled. If too many mobile device users simultaneously choose to offload the computation
tasks to the cloud via the same wireless channel, they may incur severe
interference and  this would lead to a low data rate. When the data rate of a mobile device user $n$  is low, it would consume high energy in the wireless access for offloading the computation input data to
cloud and incur long transmission time as well. In this case, it would be more beneficial for the user to compute the task locally on the mobile device to avoid the long processing time and high energy consumption by the cloud computing approach. Based on this insight, we first define the concept of beneficial cloud computing.
\begin{definition}\label{defn:beneficial}
Given a computation offloading decision
profile $\boldsymbol{a}$, the decision
$a_{n}$ of user $n$ that chooses the cloud computing approach (i.e.,
$a_{n}>0$) is \textbf{beneficial} if the cloud computing approach does not
incur higher overhead than the local computing approach (i.e., $K_{n}^{c}(\boldsymbol{a})\leq K_{n}^{m}$).
\end{definition}

The concept of beneficial cloud computing plays an important role in the mobile-edge cloud computing. On the one hand, from the user's perspective, beneficial cloud computing ensures the individual rationality, i.e., a mobile device user would not suffer performance loss by adopting the cloud computing approach. On the other hand, from the telecom operator's point of view, the larger number of users achieving beneficial cloud computing implies a higher utilization ratio of the cloud resources and a higher revenue of providing mobile-edge cloud computing service. \rev{Thus, different from traditional multi-user traffic scheduling problem, when determining the wireless access schedule for computation offloading, we need to ensure that for a user choosing cloud computing, that user must be a beneficial cloud computing user. Otherwise, the user will not follow the computation offloading schedule, since it can switch to the local computing approach to reduce the computation overhead.}

\subsection{Finding Centralized Optimum is NP-Hard}

\rev{We first consider the centralized optimization problem in term of the performance metric of the total number of beneficial cloud computing users. We will further consider another important metric of the system-wide computation overhead later.} Mathematically, we can model
the problem as follows:
\begin{eqnarray}
\max_{\boldsymbol{a}} & \sum_{n\in\mathcal{N}}I_{\{a_{n}>0\}}\label{eq:M1}\\
\mbox{subject to} & K_{n}^{c}(\boldsymbol{a})\leq K_{n}^{m}, & \forall a_{n}>0,n\in\mathcal{N},\nonumber \\
 & a_{n}\in\{0,1,...,M\}, & \forall n\in\mathcal{N}.\nonumber
\end{eqnarray}
Here $I_{\{A\}}$ is an indicator function with $I_{\{A\}}=1$ if
the event $A$ is true and $I_{\{A\}}=0$ otherwise.

Unfortunately, it turns out that the problem of finding the maximum number of  beneficial cloud computing users can be extremely challenging.
\begin{theorem}
\label{thm:The-problem-in}The problem in (\ref{eq:M1}) that computes
the maximum number of beneficial cloud computing users is NP-hard.
\end{theorem}

\begin{proof}
To proceed, we first introduce the maximum cardinality bin packing
problem \cite{loh2009solving}: we are given $N$ items with sizes $p_{i}$ for $i\in\mathcal{N}$
and $M$ bins of identical capacity $C$, and the objective is to
assign a maximum number of items to the fixed number of bins without
violating the capacity constraint. Mathematically, we can formulate
the problem as
\begin{eqnarray}
\max & \sum_{i=1}^{N}\sum_{j=1}^{M}x_{ij} & \label{eq:M1-1}\\
\mbox{subject to} & \sum_{i=1}^{N}p_{i}x_{ij}\leq C, & \forall j\in\mathcal{M},\nonumber \\
 & \sum_{j=1}^{M}x_{ij}\leq1, & \forall i\in\mathcal{N},\nonumber \\
 & x_{ij}\in\{0,1\}, & \forall i\in\mathcal{N},j\in\mathcal{M}.\nonumber
\end{eqnarray}
It is known from \cite{loh2009solving} that the maximum cardinality bin packing problem
above is NP-hard.

For our problem, according to Theorem \ref{lem:Given-the-strategies}, we know that a user $n$ that
can achieve beneficial cloud computing if and only if its received
interference $\sum_{i\in\mathcal{N}\backslash\{n\}:a_{i}=a_{n}}q_{i}g_{i,s}\leq T_{n}$.
Based on this, we can transform the maximum cardinality bin packing
problem to a special case of our problem of finding the maximum number
of beneficial cloud computing users as follows. We can regard the
items and the bins in the maximum cardinality bin packing problem
as the mobile device users and channels in our problem, respectively.
Then the size of an item $n$ and the capacity constraint of each
bin can be given as $p_{n}=q_{n}g_{n,s}$ and $C=T_{n}+q_{n}g_{n,s}$,
respectively. By this, we can ensure that as long as a user $n$ on
its assigned channel $a_{n}$ achieves the beneficial cloud computing,
for an item $n$, the total sizes of the items on its assigned bin
$a_{n}$ will not violate the capacity constraint $C$. This is due
to the fact that $\sum_{i\in\mathcal{N}\backslash\{n\}:a_{i}=a_{n}}q_{i}g_{i,s}\leq T_{n}$,
which implies that $\sum_{i=1}^{N}p_{i}x_{i,a_{n}}=\sum_{i\in\mathcal{N}\backslash\{n\}:a_{i}=a_{n}}q_{i}g_{i,s}+q_{n}g_{n,s}\leq C.$

Therefore, if we have an algorithm that can find the maximum number
of beneficial cloud computing users, then we can also obtain the optimal
solution to the maximum cardinality bin packing problem. Since
the maximum cardinality bin packing problem is NP-hard, our problem
is hence also NP-hard.
\end{proof}

The key idea of proof is to show that the maximum cardinality bin packing problem (which is known to be NP-hard \cite{loh2009solving}) can be reduced to a special case of our problem. Theorem \ref{thm:The-problem-in} provides the major
motivation for our game theoretic study, because it suggests
that the centralized optimization problem is fundamentally
difficult. By leveraging the intelligence of each individual mobile device user, game theory is a powerful tool for devising distributed mechanisms with low complexity, such that the users can self-organize into a mutually satisfactory solution. This can also help to ease the heavy burden of complex centralized computing and management by the cloud operator. Moreover, another key rationale of adopting the game theoretic approach is that the mobile devices are owned by different individuals and they may pursue different interests. Game theory is a useful framework to analyze
the interactions among multiple mobile device users who act in their own interests and devise incentive compatible computation offloading mechanisms such that no user has the incentive to deviate unilaterally.

\rev{
Besides the performance metric of the number of beneficial cloud computing  users, in this paper we also consider another important metric of the system-wide computation overhead, i.e., \begin{eqnarray}
\min_{\boldsymbol{a}} & \sum_{n\in\mathcal{N}}Z_{n}(\boldsymbol{a})\label{eq:M1111}\\
\mbox{subject to}  & a_{n}\in\{0,1,...,M\}, & \forall n\in\mathcal{N}.\nonumber
\end{eqnarray}
Note that the centralized optimization problem for minimizing the system-wide computation overhead is also NP-hard, since it involves a combinatorial optimization over the multi-dimensional discrete space (i.e., $\{0,1,...,M\}^{N}$). As shown in Sections \ref{performance} and \ref{sec:Numerical-Results}, the proposed game theoretic solution can also achieve superior performance in terms of the performance metric of the system-wide computation overhead.
}

\subsection{Game Formulation}

We then consider the distributed computation offloading decision making problem among
the mobile device users. Let $a_{-n}=(a_{1},...,a_{n-1},a_{n+1},...,a_{N})$
be the computation offloading decisions by all other users except
user $n$. Given other users' decisions $a_{-n}$, user $n$
would like to select a proper decision $a_{n}$, by using either the local
computing ($a_{n}=0$) or the cloud computing via a wireless channel ($a_{n}>0$) to minimize its computation overhead, i.e.,
\[
\min_{a_{n}\in\mathcal{A}_{n}\triangleq\{0,1,...,M\}}Z_{n}(a_{n},a_{-n}),\forall n\in\mathcal{N}.
\]
 According to (\ref{eq:l3}) and (\ref{eq:c4}), we can obtain the
overhead function of mobile device user $n$ as
\begin{equation}
Z_{n}(a_{n},a_{-n})=\begin{cases}
K_{n}^{m}, & \mbox{if }a_{n}=0,\\
K_{n}^{c}(\boldsymbol{a}), & \mbox{if }a_{n}>0.
\end{cases}\label{eq:V1}
\end{equation}

We then formulate the problem above as a strategic game $\Gamma=(\mathcal{N},\{\mathcal{A}_{n}\}_{n\in\mathcal{N}},\{Z_{n}\}_{n\in\mathcal{N}})$,
where the set of mobile device users $\mathcal{N}$ is the set of players,
$\mathcal{A}_{n}$ is the set of strategies for player
$n$, and the overhead function $Z_{n}(a_{n},a_{-n})$ of each user
$n$ is the cost function to be minimized by player $n$. In the sequel,
we call the game $\Gamma$ as the multi-user computation offloading game.
We now introduce the important concept of Nash equilibrium.
\begin{definition}
A strategy profile $\boldsymbol{a}^{*}=(a_{1}^{*},...,a_{N}^{*})$
is a \textbf{Nash equilibrium} of the multi-user computation offloading
game if at the equilibrium $\boldsymbol{a}^{*}$, no user can further reduce its overhead by unilaterally
changing its strategy, i.e.,
\begin{equation}
Z_{n}(a_{n}^{*},a_{-n}^{*})\leq Z_{n}(a_{n},a_{-n}^{*}),\forall a_{n}\in\mathcal{A}_{n},n\in\mathcal{N}.\label{eq:ne1}
\end{equation}
\end{definition}

According to the concept of Nash equilibrium, we first have the following observation.
\begin{corollary}
For the multi-user computation offloading game, if a user $n$ at Nash equilibrium $\boldsymbol{a}^{*}$ chooses cloud computing approach (i.e., $a_{n}^{*}>0$), then the user $n$ must be a beneficial cloud computing user.
\end{corollary}

This is because if a user choosing the cloud computing approach is not a beneficial cloud computing user at the equilibrium, then the user can improve  its benefit by just switching to the local computing approach, which contradicts with the fact that no user can improve unilaterally at the Nash equilibrium. Furthermore, the Nash equilibrium also ensures the nice self-stability property such that
the users at the equilibrium can achieve a mutually satisfactory solution and no user has the incentive to deviate. This property is very important to the multi-user computation offloading problem, since the mobile devices are owned by different individuals and they may act in their own interests.

\subsection{Structural Properties}

We next study the existence of Nash equilibrium of the multi-user computation offloading game. To proceed, we shall resort to
a powerful tool of potential game \cite{monderer1996potential}.
\begin{definition}
A game is called a potential game if it admits potential function
$\Phi(\boldsymbol{a})$ such that for every $n\in\mathcal{N}$, $a_{-n}\in\prod_{i\neq n}\mathcal{A}_{i}$,
and $a_{n}^{'},a_{n}\in\mathcal{A}_{n}$, if
\begin{equation}
Z_{n}(a_{n}^{'},a_{-n})<Z_{n}(a_{n},a_{-n}),\label{eq:p1}
\end{equation}
we have
\begin{equation}
\Phi(a_{n}^{'},a_{-n})<\Phi(a_{n},a_{-n}).\label{eq:p2}
\end{equation}
\end{definition}

An appealing property of the potential game is that it always admits a Nash equilibrium  and possesses the
finite improvement property, such that any asynchronous better response
update process (i.e., no more than one player updates the strategy to reduce the overhead
at any given time) must be finite and leads to a Nash equilibrium \cite{monderer1996potential}.

To show the multi-user computation offloading game is a potential game, we first show the following result.
\begin{lemma}
\label{lem:Given-the-strategies}Given a computation offloading decision
profile $\boldsymbol{a}$, a user $n$ achieves beneficial
cloud computing if its received interference $\mu_{n}(\boldsymbol{a})\triangleq\sum_{i\in\mathcal{N}\backslash\{n\}:a_{i}=a_{n}}q_{i}g_{i,s}$
on the chosen wireless channel $a_{n}>0$ satisfies that $\mu_{n}(\boldsymbol{a})\leq T_{n}$,
with the threshold
\[
T_{n}=\frac{q_{n}g_{n,s}}{2^{\frac{\left(\lambda_{n}^{t}+\lambda_{n}^{e}q_{n}\right)b_{n}}{w\left(\lambda_{n}^{t}e_{n}^{m}+\lambda_{n}^{e}e_{n}^{m}-\lambda_{n}^{e}L_{n}-\lambda_{n}^{t}t_{n,exe}^{c}\right)}}-1}-\varpi_{0}.
\]
\end{lemma}
\begin{proof}
According to (\ref{eq:l3}), (\ref{eq:c4}), and
Definition \ref{defn:beneficial}, we know that the condition $K_{n}^{c}(\boldsymbol{a})\leq K_{n}^{m}$ is equivalent to
\[
\frac{\left(\lambda_{n}^{t}+\lambda_{n}^{e}q_{n}\right)b_{n}}{r_{n}(\boldsymbol{a})}+\lambda_{n}^{e}L_{n}+\lambda_{n}^{t}t_{n,exe}^{c}\leq\lambda_{n}^{t}t_{n}^{m}+\lambda_{n}^{e}e_{n}^{m}.
\]
That is,
\[
r_{n}(\boldsymbol{a})\geq\frac{\left(\lambda_{n}^{t}+\lambda_{n}^{e}q_{n}\right)b_{n}}{\lambda_{n}^{t}t_{n}^{m}+\lambda_{n}^{e}e_{n}^{m}-\lambda_{n}^{e}L_{n}-\lambda_{n}^{t}t_{n,exe}^{c}}.
\]
According to (\ref{eq:R1}), we then have that
\[
\sum_{i\in\mathcal{N}\backslash\{n\}:a_{i}=a_{n}}q_{i}g_{i,s}\leq\frac{q_{n}g_{n,s}}{2^{\frac{\left(\lambda_{n}^{t}+\lambda_{n}^{e}q_{n}\right)b_{n}}{w\left(\lambda_{n}^{t}e_{n}^{m}+\lambda_{n}^{e}e_{n}^{m}-\lambda_{n}^{e}L_{n}-\lambda_{n}^{t}t_{n,exe}^{c}\right)}}-1}-\varpi_{0}.
\]
\end{proof}

According to Lemma \ref{lem:Given-the-strategies}, we see that when
the received interference $\mu_{n}(\boldsymbol{a})$ of user $n$ on a wireless channel is lower enough, it is beneficial for the user  to adopt cloud computing approach and offload the computation
to the cloud. Otherwise, the user $n$ should compute the task on
the mobile device locally. Based on Lemma \ref{lem:Given-the-strategies}, we show that the multi-user computation offloading game is indeed a potential
game by constructing the potential function as
\begin{align}
\Phi(\boldsymbol{a})  = & \frac{1}{2}\sum_{i=1}^{N}\sum_{j\ne i}q_{i}g_{i,s}q_{j}q_{j,s}I_{\{a_{i}=a_{j}\}}I_{\{a_{i}>0\}}\nonumber \\
   & +\sum_{i=1}^{N}q_{i}g_{i,s}T_{i}I_{\{a_{i}=0\}}.\label{eq:p4}
\end{align}

\begin{theorem}
\label{thm:The-general-decentralized}The multi-user computation
offloading game is a potential game with the potential function
as given in (\ref{eq:p4}), and hence always has a Nash equilibrium
and the finite improvement property.\end{theorem}
\begin{proof}
Suppose that a user $k\in\mathcal{N}$ updates its current decision
$a_{k}$ to the decision $a_{k}^{'}$ and this leads to a decrease
in its overhead function, i.e., $Z_{k}(a_{k},a_{-k})>Z_{k}(a_{k}^{'},a_{-k}).$
According to the definition of potential game, we will show that this
also leads to a decrease in the potential function, i.e., $\Phi(a_{k},a_{-k})>\Phi(a_{k}^{'},a_{-k})$.
We will consider the following three cases: 1) $a_{k}>0$ and $a_{k}^{'}>0$;
2) $a_{k}=0$ and $a_{k}^{'}>0$; 3) $a_{k}>0$ and $a_{k}^{'}=0$.

For case 1), since the function of $w\log_{2}(x)$ is monotonously
increasing in terms of $x$, according to (\ref{eq:R1}), we know
that the condition $Z_{k}(a_{k},a_{-k})>Z_{k}(a_{k}^{'},a_{-k})$
implies that
\begin{equation}
\sum_{i\in\mathcal{N}\backslash\{k\}:a_{i}=a_{k}}q_{i}g_{i,s}>\sum_{i\in\mathcal{N}\backslash\{k\}:a_{i}=a_{k}^{'}}q_{i}g_{i,s}.\label{eq:C11}
\end{equation}
Since $a_{k}>0$ and $a_{k}^{'}>0$, according to (\ref{eq:p4}) and
(\ref{eq:C11}), we then know that\begin{align}
   & \Phi(a_{k},a_{-k})-\Phi(a_{k}^{'},a_{-k})\nonumber\\
 = & \frac{1}{2}q_{k}g_{k,s}\sum_{i\ne k}q_{i}g_{i,s}I_{\{a_{i}=a_{k}\}} +\frac{1}{2}\sum_{k\ne i}q_{i}g_{i,s}I_{\{a_{k}=a_{i}\}}q_{k}g_{k,s}\nonumber\\
   & -\frac{1}{2}q_{k}g_{k,s}\sum_{i\ne k}q_{i}g_{i,s}I_{\{a_{i}=a_{k}^{'}\}} -\frac{1}{2}\sum_{k\ne i}q_{i}g_{i,s}I_{\{a_{k}^{'}=a_{i}\}}q_{k}g_{k,s}\nonumber\\
  = & q_{k}g_{k,s}\sum_{i\ne k}q_{i}g_{i,s}I_{\{a_{i}=a_{k}\}} -q_{k}g_{k,s}\sum_{i\ne k}q_{i}g_{i,s}I_{\{a_{i}=a_{k}^{'}\}} > 0. \label{eq:lslsls1}
\end{align}


For case 2), since $a_{k}=0$, $a_{k}^{'}>0$, and $Z_{k}(a_{k},a_{-k})>Z_{k}(a_{k}^{'},a_{-k})$,
we know that $\sum_{i\in\mathcal{N}\backslash\{k\}:a_{i}=a_{k}^{'}}q_{i}g_{i,s}<T_{k}$. This implies that\begin{align}
   & \Phi(a_{k},a_{-k})-\Phi(a_{k}^{'},a_{-k})\nonumber\\
  = & q_{k}g_{k,s}T_{k}\nonumber\\
   & -\frac{1}{2}q_{k}g_{k,s}\sum_{i\ne k}q_{i}g_{i,s}I_{\{a_{i}=a_{k}^{'}\}}-\frac{1}{2}\sum_{k\ne i}q_{i}g_{i,s}I_{\{a_{k}^{'}=a_{i}\}}q_{k}g_{k,s}\nonumber\\
  = & q_{k}g_{k,s}T_{k}-q_{k}g_{k,s}\sum_{i\ne k}q_{i}g_{i,s}I_{\{a_{i}=a_{k}^{'}\}} > 0. \label{eq:lslsls2}
\end{align}

For case 3), by the similar argument in case 2), when $a_{k}>0$ and $a_{k}^{'}=0$,  we can also show that $Z_{k}(a_{k},a_{-k})>Z_{k}(a_{k}^{'},a_{-k})$ implies $\Phi(a_{k},a_{-k})>\Phi(a_{k}^{'},a_{-k})$.

Combining results in the three cases above, we can hence conclude that the multi-user computation offloading game is a potential game.
\end{proof}

The key idea of the proof is to show that when a user $k\in\mathcal{N}$ updates its current decision
$a_{k}$ to a better decision $a_{k}^{'}$, the decrease in its overhead function will lead to the decrease in the potential function of the multi-user computation offloading game. Theorem \ref{thm:The-general-decentralized} implies that any asynchronous
better response update process is guaranteed to reach a Nash equilibrium within
a finite number of iterations. We shall exploit such finite improvement property for the distributed computation offloading algorithm design
in following Section \ref{sec:Decentralized-Computation-Offloa-1}.

\section{Distributed Computation Offloading Algorithm}\label{sec:Decentralized-Computation-Offloa-1}

In this section we develop a distributed computation offloading algorithm in Algorithm \ref{alg:Decentralized-computation-offloa}
for achieving the Nash equilibrium of the multi-user computation offloading game.

\begin{algorithm}[tt]
\begin{algorithmic}[1]
\State \textbf{initialization:}
\State each mobile device user $n$ \textbf{chooses} the computation decision $a_{n}(0)=0$.
\State \textbf{end initialization\newline}

\Repeat{ for each user $n$ and each decision slot $t$ in parallel:}
        \State \textbf{transmit} the pilot signal  on the chosen channel $a_{n}(t)$ to the wireless base-station $s$.
        \State \textbf{receive} the information of the received  powers on all the channels from the wireless base-station $s$.
        \State \textbf{compute}  the best response set $\Delta_{n}(t)$.
        \If{ $\Delta_{n}(t)\neq\varnothing$}
            \State \textbf{send} RTU message to the cloud for contending for the decision update opportunity.
            \If{ \textbf{receive} the UP message from the cloud}
                \State \textbf{choose}  the decision $a_{n}(t+1)\in\Delta_{n}(t)$ for next slot.
            \Else{ choose the original decision $a_{n}(t+1)=a_{n}(t)$ for next slot.}
            \EndIf
        \Else{ \textbf{choose} the original decision $a_{n}(t+1)=a_{n}(t)$ for next slot.}
        \EndIf
\Until{END message is received from the cloud}

\end{algorithmic}
\caption{\label{alg:Decentralized-computation-offloa}Distributed Computation Offloading Algorithm}
\end{algorithm}

\subsection{Algorithm Design}
The motivation of using the distributed computation offloading algorithm is to enable mobile device
users to achieve a mutually satisfactory decision making, prior to the computation task execution. The key idea of the algorithm design is to utilize the
finite improvement property of the multi-user computation offloading game and let one mobile device user improve
its computation offloading decision at a time. Specifically, by using the clock signal from the wireless base-station for synchronization, we consider a slotted time structure for the computation offloading decision update. Each decision slot $t$ consists the following two stages:

(1) \textbf{Wireless Interference Measurement}: at this stage, we measure the interference on different channels for wireless access. Specifically, each mobile device user $n$ who selects decision $a_{n}(t)>0$ (i.e., cloud computing approach) at the current decision slot will transmit some pilot signal on its chosen channel $a_{n}(t)$ to the wireless base-station $s$. The wireless base-station then measures the total received  power $\rho_{m}(\boldsymbol{a}(t))\triangleq\sum_{i\in\mathcal{N}:a_{i}(t)=m}q_{i}g_{i,s}$ on each channel $m\in\mathcal{M}$ and feedbacks the information of the received  powers on all the channels (i.e., $\{\rho_{m}(\boldsymbol{a}(t)),m\in\mathcal{M}\}$) to the mobile device users. Accordingly, each user $n$ can obtain its  received interference $\mu_{n}(m,a_{-n}(t))$ from other users on each channel $m\in\mathcal{M}$ as
    \[
\mu_{n}(m,a_{-n}(t))=\begin{cases}
\rho_{m}(\boldsymbol{a}(t))-q_{n}g_{n,s}, & \mbox{if }a_{n}(t)=m,\\
\rho_{m}(\boldsymbol{a}(t)), & \mbox{otherwise.}
\end{cases}
\]
That is, for its current chosen channel $a_{n}(t)$, user $n$ determines the received interference by subtracting its own power from the total measured power; for other channels over which user $n$ does not transmit the pilot signal, the received interference is equal to the total measured power.

(2) \textbf{Offloading Decision Update}: at this stage, we exploit the finite improvement
property of the multi-user computation offloading game by having one mobile device user carry out a decision update. Based on the information of the measured interferences $\{\mu_{n}(m,a_{-n}(t)),m\in\mathcal{M}\}$ on different channels, each mobile device user $n$ first computes its set of best response update as
\begin{eqnarray*}
\Delta_{n}(t) & \triangleq & \{\tilde{a}:\tilde{a}=\arg\min_{a\in\mathcal{A}_{n}}Z_{n}(a,a_{-n}(t)) \mbox{ and }\\
 &  & Z_{n}(\tilde{a},a_{-n}(t))<Z_{n}(a_{n}(t),a_{-n}(t))\}.
\end{eqnarray*}
Then, if $\Delta_{n}(t)\neq\varnothing$ (i.e., user $n$ can improve its decision),
user $n$ will send a request-to-update (RTU) message to the cloud to indicate that it wants to contend for the decision update opportunity. Otherwise, user $n$ will not contend and  adhere to the current decision at next decision slot, i.e., $a_{n}(t+1)=a_{n}(t)$. Next, the cloud will randomly select one user $k$ out of the set of users who have sent the RTU messages and send the update-permission (UP) message to the user $k$ for updating its decision for the next slot as $a_{n}(t+1)\in\Delta_{n}(t)$. For other users who do not receive the UP message from the cloud, they will not update their decisions and choose the same decisions at next slot, i.e., $a_{n}(t+1)=a_{n}(t)$.

\subsection{Convergence Analysis}
According to the finite improvement property in Theorem \ref{thm:The-general-decentralized},
the algorithm will converge to a Nash equilibrium of the multi-user
computation offloading game within finite number of decision slots. In practice, we can implement that the computation offloading decision update process terminates when no RTU messages are received by the cloud. In this case, the cloud will broadcast the END message to all the mobile device users and each user will execute the computation task according to the decision obtained at the last decision slot by the algorithm. Due to the property of Nash equilibrium, no user has the incentive to deviate from the achieved decisions.

We then analyze the computational complexity of the distributed computation offloading algorithm. In each decision slot, each mobile device user will in parallel execute the operations in Lines $5$--$15$ of Algorithm \ref{alg:Decentralized-computation-offloa}. Since most operations only involve some basic arithmetical calculations, the dominating part is the computing of the best response update in Line $11$, which involves the sorting operation over $M$ channel measurement data and typically has a complexity of $\mathcal{O}(M\log M)$. The computational complexity in each decision slot is hence $\mathcal{O}(M\log M)$. Suppose that it takes $C$ decision slots for the algorithm to terminate. Then the total computational complexity of the distributed computation offloading algorithm is $\mathcal{O}(CM\log M)$. Let $T_{max}\triangleq\max_{n\in\mathcal{N}}\{T_{n}\},$ $Q_{n}\triangleq q_{n}g_{n,s},$ $Q_{max}\triangleq\max_{n\in\mathcal{N}}\{Q_{n}\},$ and $Q_{min}\triangleq\min_{n\in\mathcal{N}}\{Q_{n}\}$. For the number of decision slots $C$ for convergence, we have the following result.

\begin{theorem}\label{thm:convergence}
When $T_{n}$ and $Q_{n}$ are non-negative integers for any $n\in\mathcal{N}$,
the distributed computation offloading algorithm will terminate within
at most $\frac{Q_{max}^{2}}{2Q_{min}}N^{2}+\frac{T_{max}Q_{max}}{Q_{min}}N$
decision slots, i.e., $C\leq\frac{Q_{max}^{2}}{2Q_{min}}N^{2}+\frac{Q_{max}T_{max}}{Q_{min}}N$. \end{theorem}
\begin{proof}
First of all, according to (\ref{eq:p4}), we know that
\begin{eqnarray}
0\leq\Phi(\boldsymbol{a}) & \leq & \frac{1}{2}\sum_{i=1}^{N}\sum_{j=1}^{N}Q_{max}^{2}+\sum_{i=1}^{N}Q_{max}T_{max}\nonumber \\
 & = & \frac{1}{2}Q_{max}^{2}N^{2}+Q_{max}T_{max}N.\label{eq:lsls2}
\end{eqnarray}

During a decision slot, suppose that a user $k\in\mathcal{N}$ updates
its current decision $a_{k}$ to the decision $a_{k}^{'}$ and this
leads to a decrease in its overhead function, i.e., $Z_{k}(a_{k},a_{-k})>Z_{k}(a_{k}^{'},a-_{k})$.
According to the definition of potential game, we will show that this
also leads to a decrease in the potential function by at least $Q_{min}$,
i.e.,
\begin{equation}
\Phi(a_{k},a_{-k})\geq\Phi(a_{k}^{'},a-_{k})+Q_{min}.\label{eq:lsls3}
\end{equation}
We will consider the following three cases: 1) $a_{k}>0$ and $a_{k}^{'}>0$;
2)$a_{k}=0$ and $a_{k}^{'}>0$; 3) $a_{k}>0$ and $a_{k}^{'}=0$.

For case 1), according to (\ref{eq:lslsls1}) in the proof of Theorem \ref{thm:The-general-decentralized}, we know that
\begin{eqnarray}
 &  & \Phi(a_{k},a_{-k})-\Phi(a_{k}^{'},a-_{k})\nonumber \\
 & = & Q_{k}\left(\sum_{i\neq k}Q_{i}I_{\{a_{i}=a_{k}\}}-\sum_{i\neq k}Q_{i}I_{\{a_{i}=a_{k}^{'}\}}\right)>0.\label{eq:lsls1}
\end{eqnarray}
Since $Q_{i}$ are integers for any $i\in\mathcal{N}$, we know that
\[\sum_{i\neq k}Q_{i}I_{\{a_{i}=a_{k}\}}\geq\sum_{i\neq k}Q_{i}I_{\{a_{i}=a_{k}^{'}\}}+1.\]
Thus, according to (\ref{eq:lsls1}), we have
\[\Phi(a_{k},a_{-k})  \geq  \Phi(a_{k}^{'},a-_{k})+Q_{k}\geq \Phi(a_{k}^{'},a-_{k})+Q_{min}.\]

For case 2), according to (\ref{eq:lslsls2}) in the proof of Theorem \ref{thm:The-general-decentralized}, we know that
\[\Phi(a_{k},a_{-k})-\Phi(a_{k}^{'},a-_{k})= Q_{k}\left(T_{k}-\sum_{i\neq k}Q_{i}I_{\{a_{i}=a_{k}^{'}\}}\right)>0.\]
By the similar augment as in case 1), we have
\[\Phi(a_{k},a_{-k})  \geq  \Phi(a_{k}^{'},a-_{k})+Q_{k}\geq \Phi(a_{k}^{'},a-_{k})+Q_{min}.\]

For case 3), by the similar argument in case 2), we can also show
that $\Phi(a_{k},a_{-k})\geq\Phi(a_{k}^{'},a-_{k})+Q_{min}.$

Thus, according to (\ref{eq:lsls2}) and (\ref{eq:lsls3}), we know
that the algorithm will terminate by driving the potential function
$\Phi(\boldsymbol{a})$ to a minimal point within at most $\frac{Q_{max}^{2}}{2Q_{min}}N^{2}+\frac{Q_{max}T_{max}}{Q_{min}}N$
decision slots. \end{proof}

Theorem \ref{thm:convergence} shows that under mild conditions the distributed computation offloading algorithm can converge in a fast manner with at most a quadratic convergence time (i.e., upper bound). Note that in practice the transmission power and  channel gain are non-negative (i.e., $q_{n}, g_{n,s}\geq 0$), we hence have $Q_{n}=\{q_{n}g_{n,s}\}\geq 0$. The non-negative condition of $T_{n}\geq 0$ ensures that a user could have the chances to achieve beneficial cloud computing (otherwise, the user should always choose the local computing). For ease of exposition, we consider that $Q_{n}$ and $T_{n}$ are integers, which can also provide a good approximation for the general case that $Q_{n}$ and $T_{n}$ could be real number. For the general case, numerical results in Section \ref{sec:Numerical-Results}  demonstrate that the distributed computation offloading algorithm can also converge in a fast manner with the number of decision slots for convergence increasing (almost) linearly with the number of users $N$. \rev{Since the time length of a slot in wireless systems is typically at time scale of microseconds (e.g., the length of a slot is around $70$ microseconds in LTE system \cite{innovations2010lte}), this implies that the time for the computation offloading decision update process is very short and can be neglectable, compared with the computation execution process, which is typically at the time scale of millisecond/seconds (e.g., for mobile gaming application, the execution time is typically several hundred milliseconds \cite{dey2013addressing}).}

%
%
%
%
%
%
%
%
%
%
%
%

\section{Performance Analysis}\label{performance}

We then analyze the performance of the distributed
computation offloading algorithm. Following the definition of price
of anarchy (PoA) in game theory \cite{roughgarden2005selfish}, we will quantify the efficiency
ratio of the worst-case Nash equilibrium over the centralized optimal
solutions in terms of two important metrics: the number of beneficial cloud computing users and the system-wide computation overhead.

\subsection{Metric I: Number of Beneficial Cloud Computing Users}
We first study the PoA in terms of the metric of the number of beneficial cloud computing users in the system. Let $\Upsilon$ be the set of Nash equilibria of the multi-user
computation offloading game and $\boldsymbol{a}^{*}=(a_{1}^{*},...,a_{N}^{*})$
denote the the centralized optimal solution that maximizes the number
of beneficial cloud computing users. Then the PoA is defined as
\[
\mbox{PoA}=\frac{\min_{\boldsymbol{a}\in\Upsilon}\sum_{n\in\mathcal{N}}I_{\{a_{n}>0\}}}{\sum_{n\in\mathcal{N}}I_{\{a_{n}^{*}>0\}}}.
\]
For the metric of the number of beneficial cloud computing users,
a larger PoA implies a better performance of the multi-user computation
offloading game solution. Recall that $T_{max}\triangleq\max_{n\in\mathcal{N}}\{T_{n}\},$
$T_{min}\triangleq\min_{n\in\mathcal{N}}\{T_{n}\},$ $Q_{max}\triangleq\max_{n\in\mathcal{N}}\{q_{n}g_{n,s}\},$
and $Q_{min}\triangleq\min_{n\in\mathcal{N}}\{q_{n}g_{n,s}\}$. We can show the following result.
\begin{theorem}\label{thm:PoA1}
Consider the multi-user computation offloading game, where $T_{n}\geq0$
for each user $n\in\mathcal{N}$. The PoA for the metric of the number
of beneficial cloud computing users satisfies that
\[
1\geq\mbox{PoA}\geq\frac{\left\lfloor \frac{T_{min}}{Q_{max}}\right\rfloor }{\left\lfloor \frac{T_{max}}{Q_{min}}\right\rfloor +1}.
\]
\end{theorem}
\begin{proof}
Let $\tilde{\boldsymbol{a}}\in\Upsilon$ be an arbitrary Nash equilibrium
of the game. Since the centralized optimum $\boldsymbol{a}^{*}$ maximizes
the number of beneficial cloud computing users, we hence have that
$\sum_{n\in\mathcal{N}}I_{\{\tilde{a}_{n}>0\}}\leq\sum_{n\in\mathcal{N}}I_{\{a_{n}^{*}>0\}}$
and PoA$\leq1$. Moreover, if $\sum_{n\in\mathcal{N}}I_{\{\tilde{a}_{n}>0\}}=N$,
we have $\sum_{n\in\mathcal{N}}I_{\{a_{n}^{*}>0\}}=N$ and PoA$=1$.
In following proof, we will focus on the case that $\sum_{n\in\mathcal{N}}I_{\{\tilde{a}_{n}>0\}}<N.$

First, we show that for the centralized optimum $\boldsymbol{a}^{*}$,
we have $\sum_{n\in\mathcal{N}}I_{\{a_{n}^{*}>0\}}\leq M\left(\left\lfloor \frac{T_{max}}{Q_{min}}\right\rfloor +1\right)$,
where $M$ is the number of channels. To proceed, we first denote
$C_{m}(\boldsymbol{a})\triangleq\sum_{i=1}^{N}I_{\{a_{i}=m\}}$ as
the number of users on channel $m$ for a given decision profile $\boldsymbol{a}$.
Since $T_{n}\geq0,$ we have $K_{n}^{c}(a_{n},a_{-n}=\boldsymbol{0})\geq K_{n}^{m}$ for $a_{n}>0$,
i.e., there exists at least a user that can achieve beneficial cloud
computing by letting the user choose cloud computing $a_{n}$ and the other users choose local computing. This implies that for the centralized optimum $\boldsymbol{a}^{*}$,
we have $\sum_{n\in\mathcal{N}}I_{\{a_{n}^{*}>0\}}\geq1$. Let $C_{m^{*}}(\boldsymbol{a}^{*})=\max_{m\in\mathcal{M}}\{C_{m}(\boldsymbol{a}^{*})\}$,
i.e., channel $m^{*}$ is the one with most users. Suppose user $n$
is on the channel $m^{*}$. Then we know that
\[
\sum_{i\in\mathcal{N}\backslash\{n\}:a_{i}=m^{*}}q_{i}g_{i,s}\leq T_{n},
\]
which implies that
\begin{eqnarray*}
\left(C_{m^{*}}(\boldsymbol{a}^{*})-1\right)Q_{min} & \leq & \sum_{i\in\mathcal{N}\backslash\{n\}:a_{i}=m^{*}}q_{i}g_{i,s}\\
 & \leq & T_{n}\leq T_{max}.
\end{eqnarray*}
It follows that
\begin{eqnarray*}
C_{m^{*}}(\boldsymbol{a}^{*}) & \leq & \left\lfloor \frac{T_{max}}{Q_{min}}\right\rfloor +1.
\end{eqnarray*}
We hence have that
\begin{align}
\sum_{n\in\mathcal{N}}I_{\{a_{n}^{*}>0\}} = & \sum_{m=1}^{M}C_{m}(\boldsymbol{a}^{*})\leq MC_{m^{*}}(\boldsymbol{a}^{*})  \\
\leq & M\left(\left\lfloor \frac{T_{max}}{Q_{min}}\right\rfloor +1\right).\label{eq:PoA1}
\end{align}

Second, for the Nash equilibrium $\tilde{\boldsymbol{a}}$, since
$\sum_{n\in\mathcal{N}}I_{\{\tilde{a}_{n}>0\}}<N$, there exists at
lease one user $\tilde{n}$ that chooses the local computing approach,
i.e., $a_{\tilde{n}}=0.$ Since $\tilde{\boldsymbol{a}}$ is a Nash
equilibrium, we have that user $\tilde{n}$ cannot reduce its overhead
by choosing computation offloading via any channel $m\in\mathcal{M}$.
We then know that
\[
\sum_{i\in\mathcal{N}\backslash\{\tilde{n}\}:\tilde{a}_{i}=m}q_{i}g_{i,s}\geq T_{\tilde{n}},\forall m\in\mathcal{M},
\]
which implies that
\begin{eqnarray*}
C_{m}(\tilde{\boldsymbol{a}})Q_{max} & \geq & \sum_{i\in\mathcal{N}\backslash\{\tilde{n}\}:\tilde{a}_{i}=m}q_{i}g_{i,s}.\\
 & \geq & T_{\tilde{n}}\geq T_{min}.
\end{eqnarray*}
It follows that
\begin{eqnarray*}
C_{m^{*}}(\tilde{\boldsymbol{a}}) & \geq & \frac{T_{min}}{Q_{max}}\geq \left\lfloor \frac{T_{min}}{Q_{max}}\right\rfloor .
\end{eqnarray*}
Thus, we have
\begin{equation}
\sum_{n\in\mathcal{N}}I_{\{\tilde{a}_{n}>0\}}=\sum_{m=1}^{M}C_{m}(\tilde{\boldsymbol{a}})\geq M\left\lfloor \frac{T_{min}}{Q_{max}}\right\rfloor .\label{eq:PoA2}
\end{equation}
Based on (\ref{eq:PoA1}) and (\ref{eq:PoA2}), we can conclude that
$\mbox{PoA}\geq\frac{\left\lfloor \frac{T_{min}}{Q_{max}}\right\rfloor }{\left\lfloor \frac{T_{max}}{Q_{min}}\right\rfloor +1},$
which completes the proof.
\end{proof}

Recall that the constraint $T_{n}\geq0$ ensures that some user can achieve beneficial cloud computing in the centralized optimum, and avoid the possibility of the PoA involving ``division by zero". Theorem \ref{thm:PoA1} implies that the worst-case performance of the  Nash equilibrium will be close
to the centralized optimum $\boldsymbol{a}^{*}$ when the gap between the best and worst users in terms of wireless access performance $q_{n},g_{n,s}$ and interference tolerance threshold $T_{n}$ for achieving beneficial cloud computing is not large.

\subsection{Metric II: System-wide Computation Overhead}
We then study the PoA in terms of another metric of the total computation overhead of all the mobile device users in the system, i.e., $\sum_{n\in\mathcal{N}}Z_{n}(\boldsymbol{a})$. Let $\bar{\boldsymbol{a}}$ be the centralized optimal solution that minimizes the system-wide computation overhead, i.e., $\bar{\boldsymbol{a}}=\arg \min_{\boldsymbol{a}\in\prod_{n=1}^{N}\mathcal{A}_{n}}\sum_{n\in\mathcal{N}}Z_{n}(\boldsymbol{a})$. Similarly, we can define the PoA as
\[
\mbox{PoA}=\frac{\max_{\boldsymbol{a}\in\Upsilon}\sum_{n\in\mathcal{N}}Z_{n}(\boldsymbol{a})}{\sum_{n\in\mathcal{N}}Z_{n}(\bar{\boldsymbol{a}})}.
\]
Note that, different from the metric of the number of beneficial cloud computing users, a smaller system-wide computation overhead is more desirable. Hence, for the metric of the system-wide computation overhead, a smaller PoA  is better. Let $K_{n,min}^{c}\triangleq\frac{\left(\lambda_{n}^{t}+\lambda_{n}^{e}q_{n}\right)b_{n}}{w\log_{2}\left(1+\frac{q_{n}g_{n,s}}{\varpi_{0}}\right)}+\lambda_{n}^{e}L_{n}+\lambda_{n}^{t}t_{n,exe}^{c}$ and $K_{n,max}^{c}\triangleq\frac{\left(\lambda_{n}^{t}+\lambda_{n}^{e}q_{n}\right)b_{n}}{w\log_{2}\left(1+\frac{q_{n}g_{n,s}}{\varpi_{0}+\left(\sum_{i\in\mathcal{N}\backslash\{n\}}q_{i}g_{i,s}\right)/M}\right)}+\lambda_{n}^{e}L_{n}+\lambda_{n}^{t}t_{n,exe}^{c}.$
We can show the following result.
\begin{theorem}
\label{thm:PoA2}For the multi-user computation offloading
game, the PoA of the metric of the system-wide computation overhead satisfies that \[ 1\leq\mbox{PoA}\leq\frac{\sum_{n=1}^{N}\min\{K_{n}^{m},K_{n,max}^{c}\}}{\sum_{n=1}^{N}\min\{K_{n}^{m},K_{n,min}^{c}\}}.\]\end{theorem}

\begin{proof}
Let $\tilde{\boldsymbol{a}}\in\Upsilon$ be an arbitrary Nash equilibrium
of the game. Since the centralized optimum $\boldsymbol{a}^{*}$ minimizes
the system-wide computation overhead, we hence first have that PoA$\geq1$.

For a Nash equilibrium  $\hat{\boldsymbol{a}}\in\Upsilon$, if $\bar{a}_{n}>0$, we shall show that the interference that a user $n$ receives from other other users on the wireless access channel $\hat{a}_{n}$ is at most \[\left(\sum_{i\in\mathcal{N}\backslash\{n\}}q_{i}g_{i,s}\right)/M.\] We prove this by contradiction. Suppose that a user $n$ at the Nash
equilibrium $\hat{\boldsymbol{a}}$ receives an interference greater
than $\left(\sum_{i\in\mathcal{N}\backslash\{n\}}q_{i}g_{i,s}\right)/M.$
Then, we have that
\begin{equation}
\sum_{i\in\mathcal{N}\backslash\{n\}:\hat{a}_{i}=\hat{a}_{n}}q_{n}g_{n,s}>\left(\sum_{i\in\mathcal{N}\backslash\{n\}}q_{i}g_{i,s}\right)/M.\label{eq:ff1}
\end{equation}
According to the property of Nash equilibrium such that no user can
improve by changing the channel unilaterally, we also have that
\begin{eqnarray*}
 &  & \sum_{i\in\mathcal{N}\backslash\{n\}:\hat{a}_{i}=m}q_{n}g_{n,s}\\
 & \geq & \sum_{i\in\mathcal{N}\backslash\{n\}:\hat{a}_{i}=\hat{a}_{n}}q_{n}g_{n,s},\forall m\in\mathcal{M}.
\end{eqnarray*}
This implies that
\begin{eqnarray}
 &  & \sum_{m=1}^{M}\sum_{i\in\mathcal{N}\backslash\{n\}:\hat{a}_{i}=m}q_{n}g_{n,s}\nonumber \\
 & \geq & M\left(\sum_{i\in\mathcal{N}\backslash\{n\}:\hat{a}_{i}=\hat{a}_{n}}q_{n}g_{n,s}\right).\label{eq:ff2}
\end{eqnarray}
According to (\ref{eq:ff1}) and (\ref{eq:ff2}), we now reach a contradiction
that
\begin{eqnarray*}
 &  & \left(\sum_{i\in\mathcal{N}\backslash\{n\}}q_{i}g_{i,s}\right)/M\\
 & < & \sum_{i\in\mathcal{N}\backslash\{n\}:\hat{a}_{i}=\hat{a}_{n}}q_{n}g_{n,s}\\
 & \leq & \left(\sum_{m=1}^{M}\sum_{i\in\mathcal{N}\backslash\{n\}:\hat{a}_{i}=m}q_{n}g_{n,s}\right)/M\\
 & \leq & \left(\sum_{i\in\mathcal{N}\backslash\{n\}}q_{i}g_{i,s}\right)/M.
\end{eqnarray*}
Thus, a user $n$ at the Nash equilibrium $\hat{\boldsymbol{a}}$
receives an interference not greater than $\left(\sum_{i\in\mathcal{N}\backslash\{n\}}q_{i}g_{i,s}\right)/M$.
Based on this, if $\hat{a}_{n}>0$, we hence have that
\[
r_{n}(\hat{\boldsymbol{a}})\geq w\log_{2}\left(1+\frac{q_{n}g_{n,s}}{\varpi_{0}+\left(\sum_{i\in\mathcal{N}\backslash\{n\}}q_{i}g_{i,s}\right)/M}\right),
\]
which implies that
\begin{align*}
 & K_{n}^{c}(\hat{\boldsymbol{a}}) \\
=& \frac{\left(\lambda_{n}^{t}+\lambda_{n}^{e}q_{n}\right)b_{n}}{r_{n}(\hat{\boldsymbol{a}})}+\lambda_{n}^{e}L_{n}+\lambda_{n}^{t}t_{n,exe}^{c}\\
    \geq & \frac{\left(\lambda_{n}^{t}+\lambda_{n}^{e}q_{n}\right)b_{n}}{w\log_{2}\left(1+\frac{q_{n}g_{n,s}}{\varpi_{0}+\left(\sum_{i\in\mathcal{N}\backslash\{n\}}q_{i}g_{i,s}\right)/M}\right)}+\lambda_{n}^{e}L_{n}+\lambda_{n}^{t}t_{n,exe}^{c}\\
= & K_{n,max}^{c}.
\end{align*}
Moreover, if $K_{n}^{m}<K_{n,max}^{c}$ and $\hat{a}_{n}>0$, then the user can always
improve by switching to the local computing approach (i.e., $\hat{a}_{n}=0$), we thus know
that \begin{align}Z_{n}(\hat{\boldsymbol{a}})\leq\min\{K_{n}^{m},K_{n,max}^{c}\}.\label{ff33}\end{align}

For the centralized optimal solution $\bar{\boldsymbol{a}}$,
if $\bar{a}_{n}>0$, we have that
\begin{eqnarray*}
r_{n}(\bar{\boldsymbol{a}}) & = & w\log_{2}\left(1+\frac{q_{n}g_{n,s}}{\varpi_{0}+\sum_{i\in\mathcal{N}\backslash\{n\}:\bar{a}_{i}=\bar{a}_{n}}q_{i}g_{i,s}}\right)\\
 &  & \leq w\log_{2}\left(1+\frac{q_{n}g_{n,s}}{\varpi_{0}}\right),
\end{eqnarray*}
 which implies that
\begin{align*}
  & K_{n}^{c}(\bar{\boldsymbol{a}}) \\
 = & \frac{\left(\lambda_{n}^{t}+\lambda_{n}^{e}q_{n}\right)b_{n}}{r_{n}(\bar{\boldsymbol{a}})}+\lambda_{n}^{e}L_{n}+\lambda_{n}^{t}t_{n,exe}^{c}\\
    \leq & \frac{\left(\lambda_{n}^{t}+\lambda_{n}^{e}q_{n}\right)b_{n}}{w\log_{2}\left(1+\frac{q_{n}g_{n,s}}{\varpi_{0}}\right)}+\lambda_{n}^{e}L_{n}+\lambda_{n}^{t}t_{n,exe}^{c}\\
 =  & K_{n,min}^{c}.
\end{align*}
Moreover, if $K_{n}^{m}<K_{n,min}^{c}$ and $\bar{a}_{n}>0$, then the system-wide computation overhead can be further reduced by letting user $n$ switch to the local computing approach (i.e., $\bar{a}_{n}=0$). This is because such a switching will not increase extra interference to other users.  We thus know
that \begin{align}Z_{n}(\bar{\boldsymbol{a}})\leq\min\{K_{n}^{m},K_{n,min}^{c}\}.\label{ff44}\end{align}

According to (\ref{ff33}) and (\ref{ff44}), we can conclude that \begin{align*}1\leq \mbox{PoA} & =  \frac{\max_{\boldsymbol{a}\in\Upsilon}\sum_{n\in\mathcal{N}}Z_{n}(\boldsymbol{a})}{\sum_{n\in\mathcal{N}}Z_{n}(\bar{\boldsymbol{a}})} \\ & \leq\frac{\sum_{n=1}^{N}\min\{K_{n}^{m},K_{n,max}^{c}\}}{\sum_{n=1}^{N}\min\{K_{n}^{m},K_{n,min}^{c}\}}.\end{align*}
\end{proof}

Intuitively, Theorem \ref{thm:PoA2} indicates that when the resource for wireless access increases
(i.e., the number of wireless access channels $M$ is larger and hence $K_{n,max}^{c}$
is smaller), the worst-case performance of  Nash equilibrium can be improved. Moreover, when users
have lower cost of local computing (i.e., $K_{n}^{m}$ is smaller), the worst-case Nash equilibrium is closer to the centralized optimum and hence
the PoA is lower.

\section{Extension to Wireless Contention Model}\label{Extension}
\rev{
In the previous sections above, we mainly focus on exploring the distributed
computation offloading problem under the wireless interference model
as given in (\ref{eq:R1}). Such wireless interference model is widely adopted in literature (see \cite{rappaport1996wireless,andrews2014will} and references therein) and can
well capture user's time average aggregate throughput in the cellular
communication scenario in which some physical layer channel access
scheme (e.g., CDMA) is adopted to allow multiple users
to share the same spectrum resource simultaneously and efficiently.
In this case, the multiple access among users for the shared spectrum
is carried out over the signal/symbol level (e.g., at the time scale
of microseconds), rather than the packet level (e.g., at the time
scale of milliseconds/seconds).

In this section, we extend our study to the wireless contention model
in which the multiple access among users for the shared spectrum is
carried out over the packet level. This is most relevant to the scenario
that some media access control protocol such as CSMA is implemented
such that users content to capture the channel for data packet transmission
for a long period (e.g., hundreds of milliseconds or several seconds)
in the WiFi-like networks (e.g., White-Space Network \cite{bahl2009white}). In this case,
we can model a user's expected throughput for computation offloading
over the chosen wireless channel $a_{n}>0$ as follows
\begin{equation}
r_{n}(\boldsymbol{a})=R_{n}\frac{W_{n}}{W_{n}+\sum_{i\in\mathcal{N}\backslash\{n\}:a_{i}=a_{n}}W_{i}},\label{eq:WC1}
\end{equation}
where $R_{n}$ is the data rate that user $n$ can achieve when it
can successfully gab the channel, and $W_{n}>0$ denotes user's weight
in the channel contention/sharing, with a larger weight $W_{n}$ implying that user $n$ is more dominant in grabbing the channel. When $W_{n}=1$ for any user $n$, it is relevant to the equal-sharing case (e.g., round robin scheduling).

Similarly, we can apply
the communication and computation models in the previous sections above to compute the overhead for
both local and cloud computing approaches, and model the distributed computation
offloading problem as a strategic game. For such multi-user computation
offloading game under the wireless contention model, we can show that
it exhibits the same structural property as the case under the wireless
interference model. We can first define the received ``interference'' (i.e., aggregated contention weights)
of user $n$ on the chosen channel as $\mu_{n}(\boldsymbol{a})=\sum_{i\in\mathcal{N}\backslash\{n\}:a_{i}=a_{n}}W_{i}$.
Then we can show the same threshold structure for the game as follow.
\begin{lemma}
\label{lem:WC1}For the multi-user computation offloading game under
the wireless contention model, a user $n$ achieves beneficial cloud
computing if its received interference $\mu_{n}(\boldsymbol{a})$
on the chosen channel $a_{n}>0$ satisfies that $\mu_{n}(\boldsymbol{a})\leq T_{n}$,
with the threshold
\[
T_{n}=\left(\frac{\left(\lambda_{n}^{t}t_{n}^{m}+\lambda_{n}^{e}e_{n}^{m}-\lambda_{n}^{e}L_{n}-\lambda_{n}^{t}t_{n,exe}^{c}\right)R_{n}}{\left(\lambda_{n}^{t}+\lambda_{n}^{e}q_{n}\right)b_{n}}-1\right)W_{n}.
\]
\end{lemma}

By exploiting the threshold structure above and following the similar arguments in the proof of Theorem \ref{thm:The-general-decentralized}, we can also show that
the multi-user computation offloading game under the wireless contention
model is a potential game.
\begin{theorem}
\label{thm:WC2}The multi-user computation offloading game under the
wireless contention model is a potential game under the wireless contention
model with the potential function as given in (\ref{eq:WC2}), and
hence always has a Nash equilibrium and the finite improvement property.
\begin{eqnarray}
\Phi(\boldsymbol{a}) & = & \frac{1}{2}\sum_{i=1}^{N}\sum_{j\ne i}W_{i}W_{j}I_{\{a_{i}=a_{j}\}}I_{\{a_{i}>o\}}\nonumber \\
 &  & +\sum_{i=1}^{N}W_{i}T_{n}I_{\{a_{n}=0\}}.\label{eq:WC2}
\end{eqnarray}
\end{theorem}
Based on Lemma \ref{lem:WC1} and Theorem (\ref{thm:WC2}), we observe that
the multi-user computation offloading game under the wireless contention
model exhibits the same structural property as the case under the
wireless interference model. Moreover, by defining $q_{n}g_{n,s}=W_{n}$,
the potential function in (\ref{eq:WC2}) is the same as that in
(\ref{eq:p4}). Thus, by regarding the aggregated contention weights $\mu_{n}(\boldsymbol{a})=\sum_{i\in\mathcal{N}\backslash\{n\}:a_{i}=a_{n}}W_{i}$ as the received interference, we can apply the distributed computation offloading algorithm
in Section \ref{sec:Decentralized-Computation-Offloa-1} to achieve the Nash equilibrium, which possesses the
same performance and convergence guarantee for the case under the wireless contention model.
}

\section{Numerical Results}\label{sec:Numerical-Results}

\begin{figure*}[tt]
\begin{minipage}[t]{0.32\linewidth}
\centering
\includegraphics[scale=0.4]{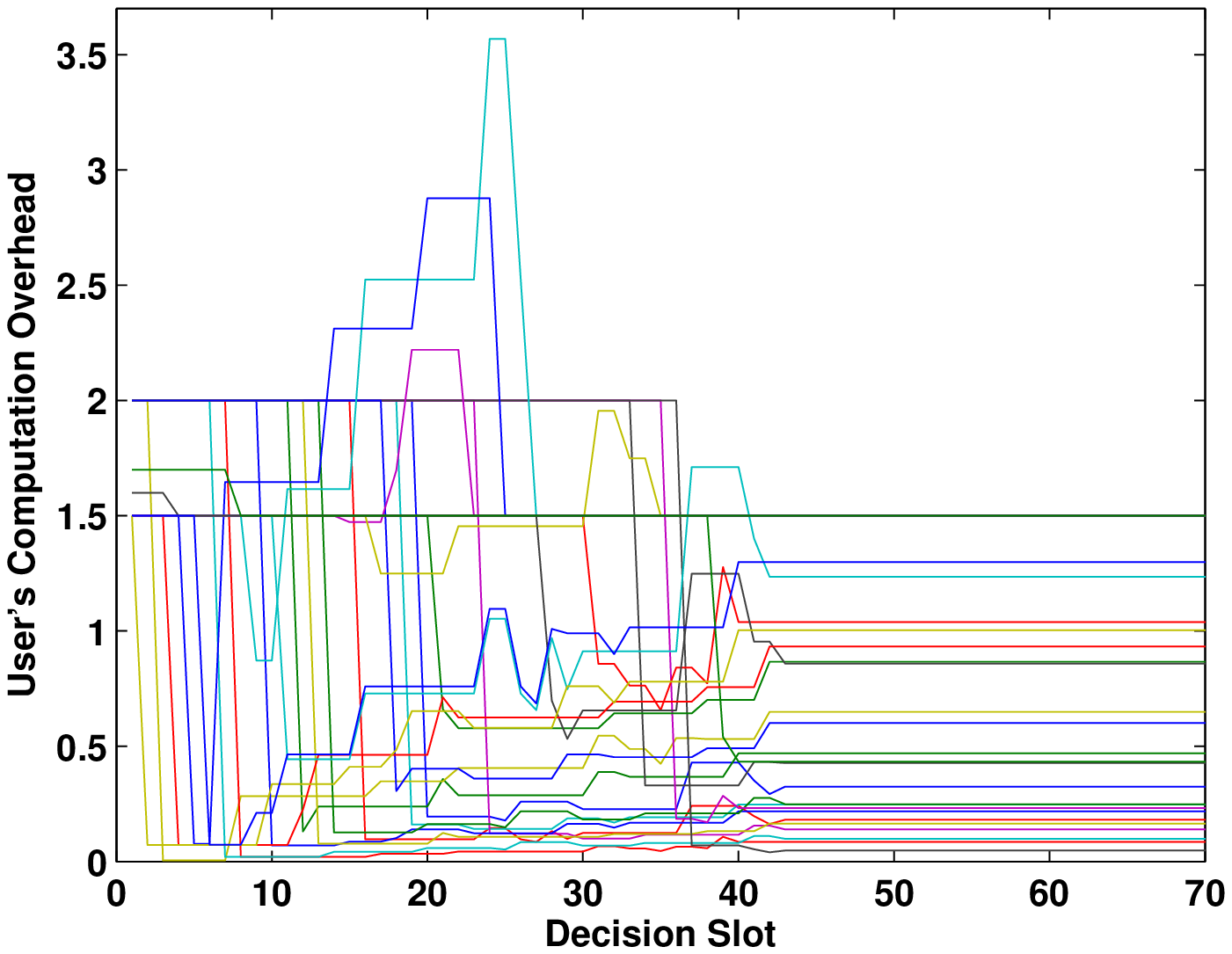}
\caption{\label{fig:UserCost} Dynamics of users' computation overhead}
\end{minipage}
\hfill
\begin{minipage}[t]{0.32\linewidth}
\centering
\includegraphics[scale=0.4]{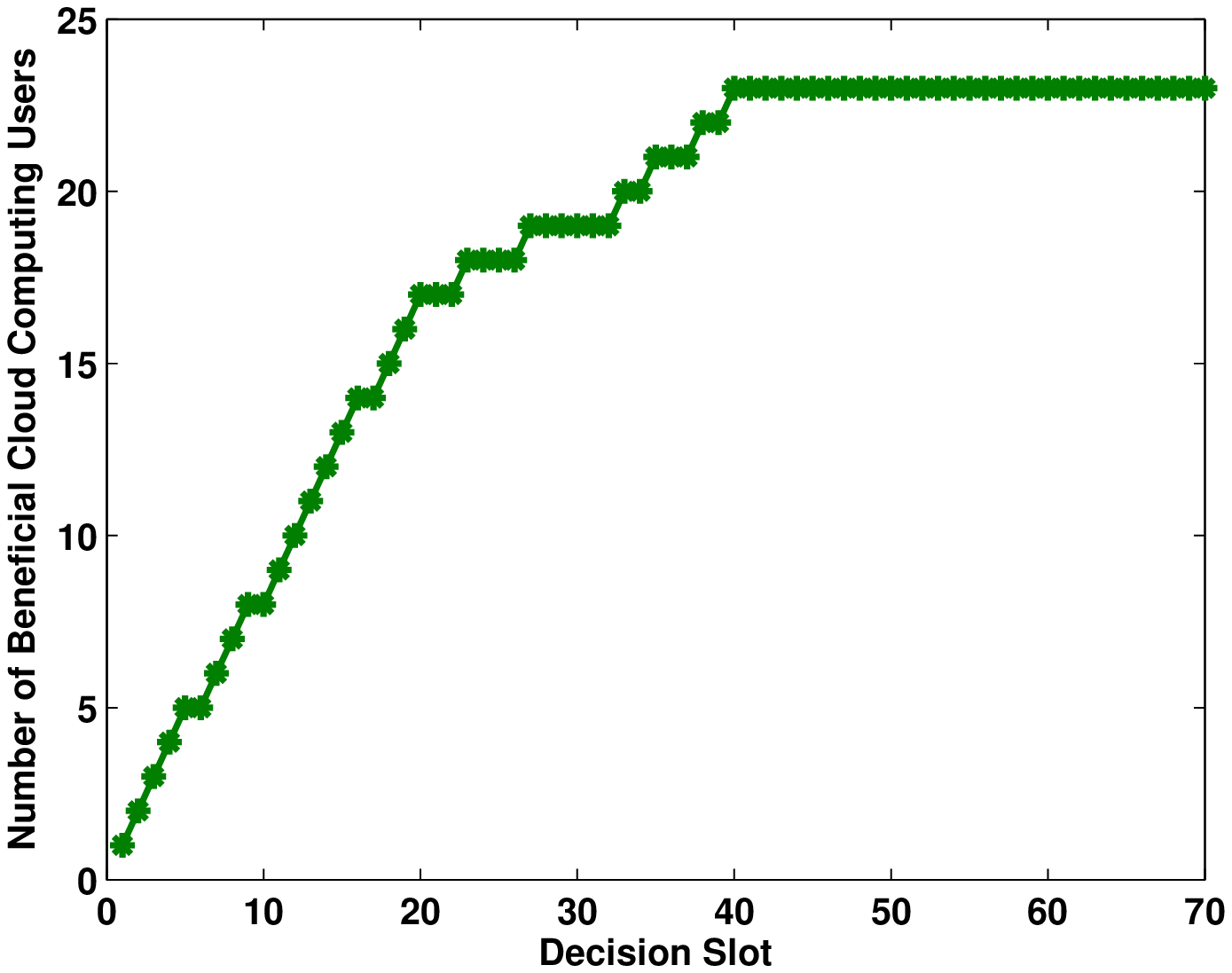}
\caption{\label{fig:BeneUser}Dynamics of the number of beneficial cloud computing users}
\end{minipage}
\hfill
\begin{minipage}[t]{0.32\linewidth}
\centering
\includegraphics[scale=0.4]{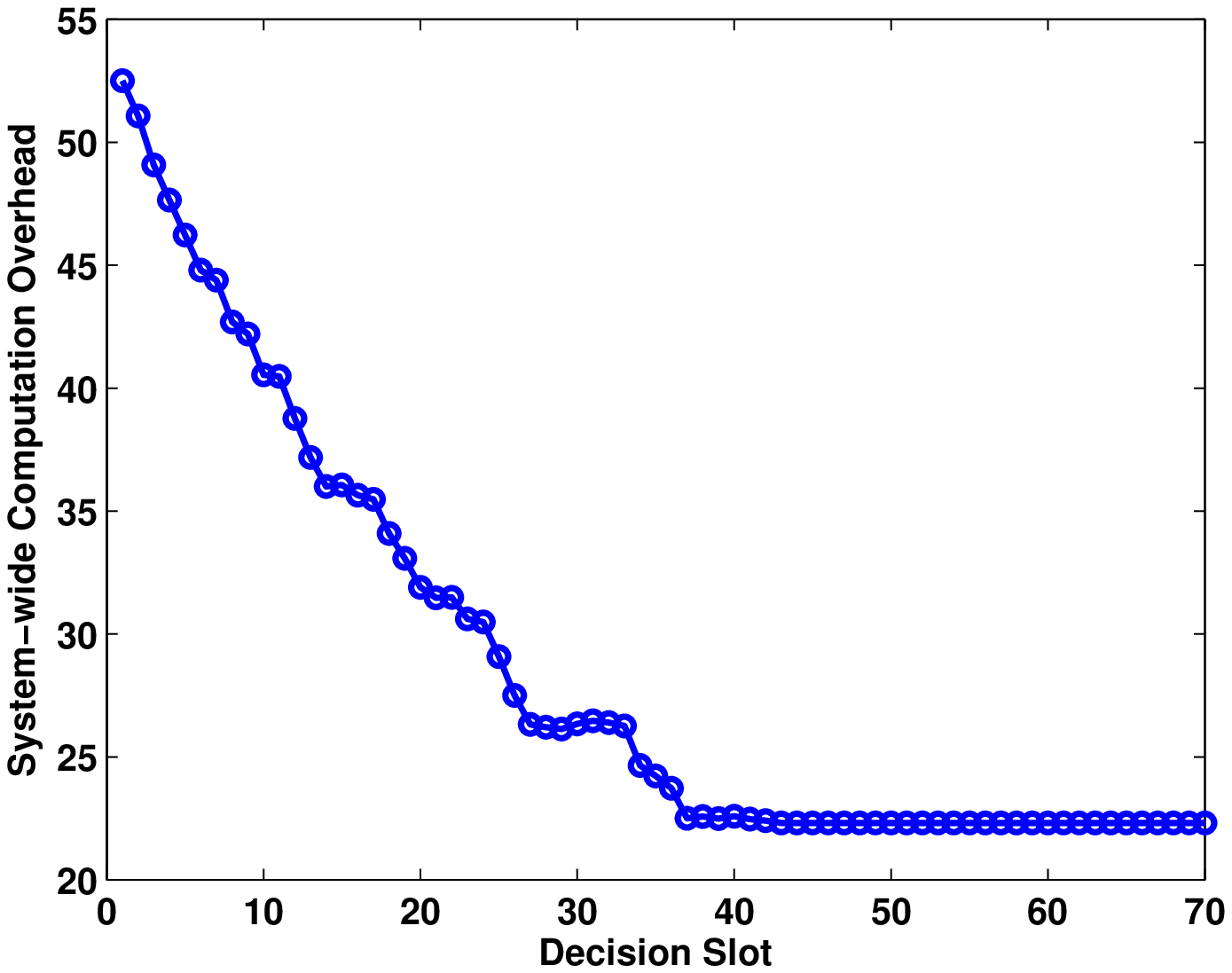}
\caption{\label{fig:SystemCost}Dynamics of system-wide computation overhead}
\end{minipage}
\end{figure*}

\begin{figure*}[tt]
\begin{minipage}[t]{0.32\linewidth}
\centering
\includegraphics[scale=0.4]{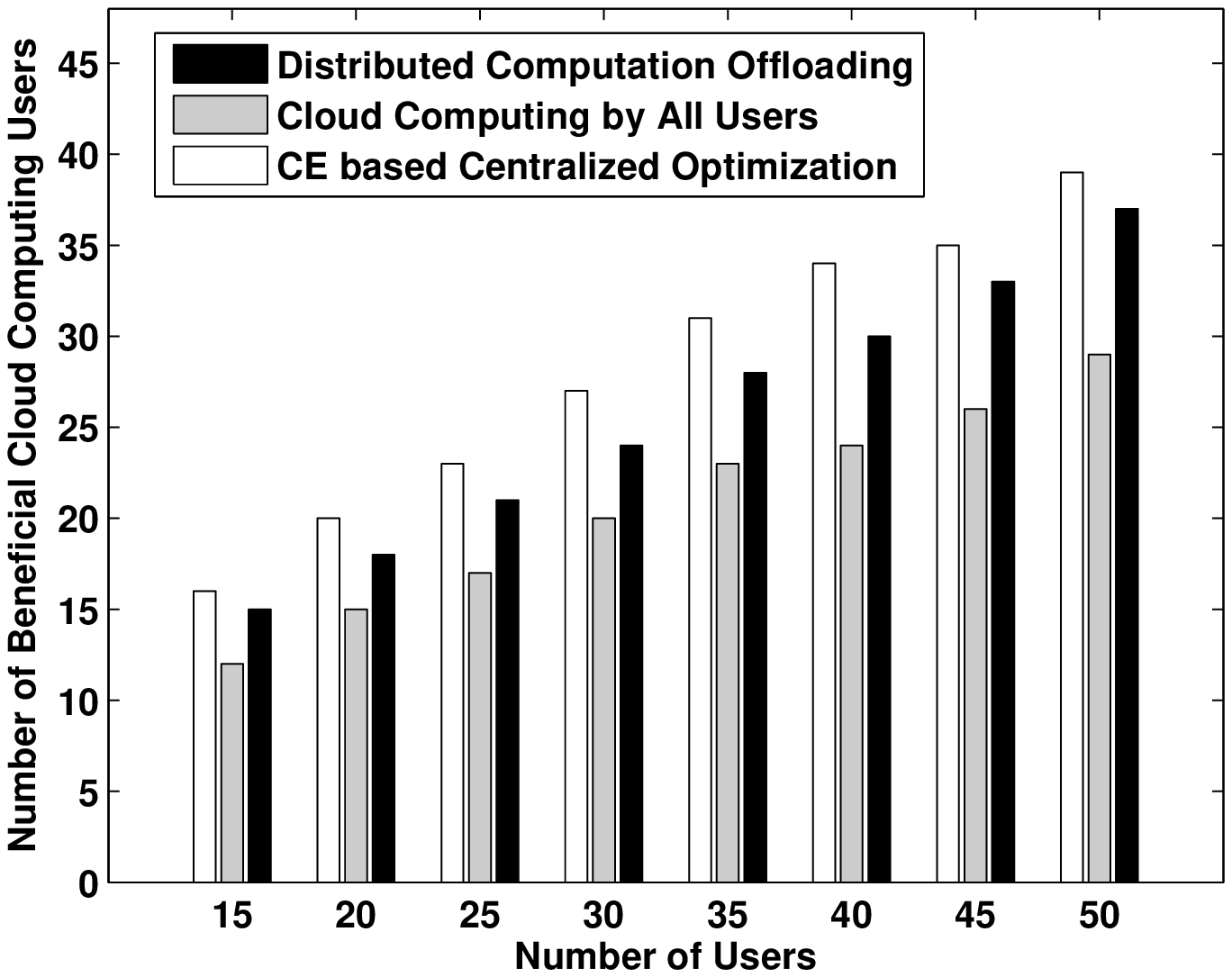}
\caption{\label{fig:Average1}Average number of beneficial cloud computing users with different number of users}
\end{minipage}
\hfill
\begin{minipage}[t]{0.32\linewidth}
\centering
\includegraphics[scale=0.4]{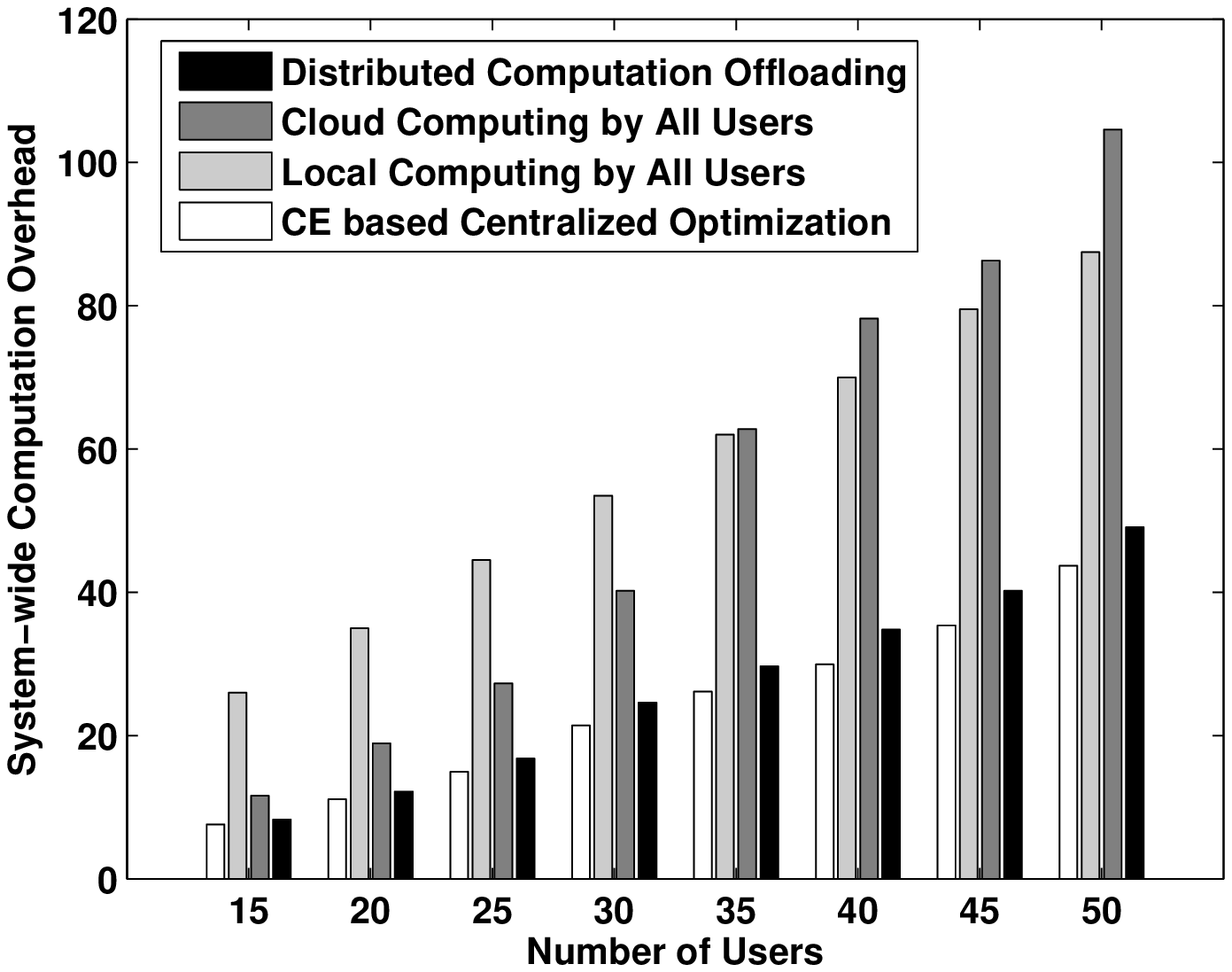}
\caption{\label{fig:Average2}Average system-wide computation overhead with different number of users}
\end{minipage}
\hfill
\begin{minipage}[t]{0.32\linewidth}
\centering
\includegraphics[scale=0.42]{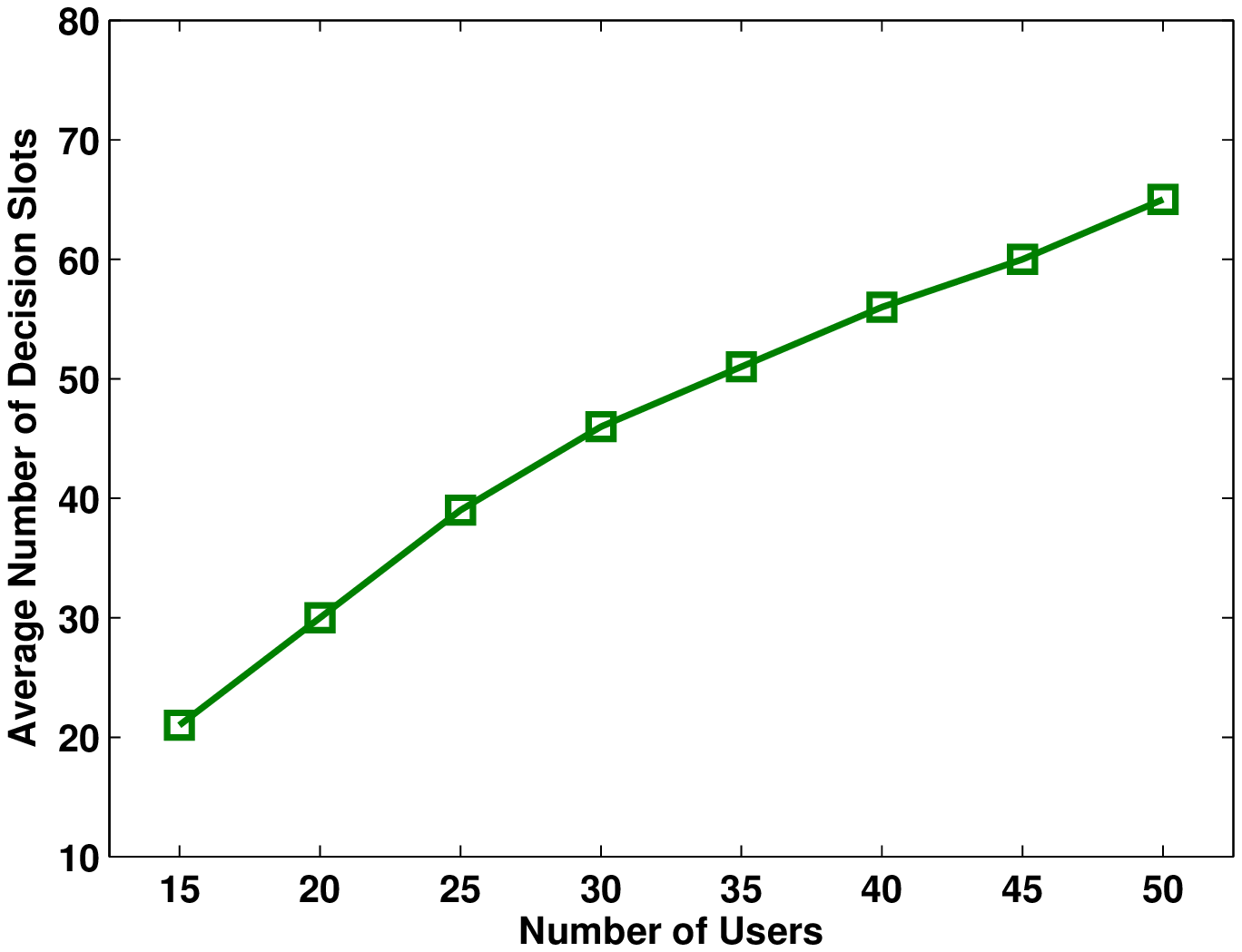}
\caption{\label{fig:Iterations}Average number of decision slots for convergence with different number of users}
\end{minipage}
\end{figure*}

In this section, we evaluate the proposed distributed computation
offloading algorithm by numerical studies. We first consider the scenario where the wireless small-cell base-station has a coverage range of $50$m \cite{quek2013small} and $N=30$ mobile device users are randomly scattered over the coverage region \cite{quek2013small}. The base-station consists of $M=5$ channels and the channel bandwidth
$w=5$ MHz. The transmission power $q_{n}=100$ mWatts and the background
noise $\varpi_{0}=-100$ dBm \cite{rappaport1996wireless}. According to the wireless interference
model for urban cellular radio environment \cite{rappaport1996wireless}, we set the channel gain $g_{n,s}=l_{n,s}^{-\alpha}$,
where $l_{n,s}$ is the distance between mobile device user $n$ and the wireless base-station and $\alpha=4$ is the path loss factor.

For the computation
task, we consider the face recognition application in \cite{soyata2012cloud},
where the data size for the computation offloading $b_{n}=5000$ KB
and the total number of CPU cycles $d_{n}=1000$ Megacycles.
The CPU computational capability $f^{m}_{n}$ of a mobile device
user $n$ is randomly assigned from the set $\{0.5,0.8,1.0\}$ GHz to account for the heterogenous computing capability of mobile devices,  
and the computational capability allocated for a user $n$ on the cloud is $f^{c}_{n}=10$ GHz \cite{soyata2012cloud}. For the decision weights of each user $n$ for both the computation time and energy, we set that $\lambda_{n}^{t}=1-\lambda_{n}^{e}$ and $\lambda_{n}^{e}$ is randomly assigned from the set $\{1, 0.5, 0\}$. In this case, if $\lambda_{n}^{e}=1$ ($\lambda_{n}^{e}=0$, respectively), a user $n$ only cares about the computation energy (computation time, respectively); if $\lambda_{n}^{e}=0.5$, then user $n$ cares both the computation time and energy.

We first show the dynamics of mobile device users' computation overhead $Z_{n}(\boldsymbol{a})$ by the
proposed distributed computation offloading algorithm in Figure
\ref{fig:UserCost}. We see that the algorithm can converge to a stable point (i.e., Nash equilibrium of the multi-user computation offloading game).  Figure \ref{fig:BeneUser} shows the dynamics of the achieved number of beneficial cloud computing users by the
proposed algorithm. It demonstrates that the algorithm can keep the number of beneficial cloud computing users in the system increasing and converge to an equilibrium. We further show the dynamics of the system-wide computation overhead $\sum_{n\in\mathcal{N}}Z_{n}(\boldsymbol{a})$ by the
proposed algorithm in Figure \ref{fig:SystemCost}. We see that the algorithm can also keep the system-wide computation overhead decreasing and converge to an equilibrium.

We then compare the distributed computation offloading
algorithm with the following solutions:

(1) \textbf{Local Computing by All Users}: each user chooses to compute its own task locally on the mobile phone.  This could correspond to the scenario that each user is risk-averse and would like to avoid any potential performance degradation due to the concurrent computation offloadings by other users.

(2) \textbf{Cloud Computing by All Users}: each user chooses to offload its own task to the cloud via a randomly selected wireless channel. This could correspond to the scenario that each user is myopic and ignores the impact of other users for cloud computing.

(3) \textbf{Cross Entropy based Centralized Optimization}: we compute the centralized optimum by the global optimization using Cross Entropy (CE) method,
which is an advanced randomized searching technique and has been shown to be efficient in finding near-optimal solutions to complex combinatorial optimization problems \cite{rubinstein2004cross}.

We run experiments with different number
of $N=15,...,50$ mobile device users \cite{quek2013small}, respectively. We repeat
each experiment $100$ times for each given user number $N$ and show the average number of beneficial cloud computing users and the average system-wide computation overhead
in Figures \ref{fig:Average1} and \ref{fig:Average2}, respectively. We see that, for the metric of the number of beneficial cloud computing users, the distributed computation offloading solution can achieve up-to $30\%$ performance improvement over the solutions by cloud computing by all users, respectively. For the metric of the system-wide computation overhead, the distributed computation offloading solution can achieve up-to $68\%$ and $55\%$, and $51\%$ overhead reduction over with the solutions by local computing by all users, and cloud computing by all users, respectively. Moreover, compared with the centralized optimal solution by CE method, the performance loss of the distributed computation offloading solution is at most $12\%$ and $14\%$, for the metrics of  number of beneficial cloud computing users and system-wide computation overhead, respectively. This demonstrates the efficiency of the proposed distributed computation offloading algorithm. Note that for the distributed computation offloading algorithm, a mobile user makes the computation offloading
decision locally based on its local parameters. While for CE based centralized optimization, the complete information is required and hence all the users need to report all their local parameters to the cloud. This would incur high system overhead for massive information collection and may raise the privacy issue as well. Moreover, since the mobile devices are owned by different individuals and they may pursue different interests, the users may not have the incentive to follow the centralized optimal solution. While, due to the property of Nash equilibrium, the distributed computation offloading solution can ensure the self-stability such that no user has the incentive to deviate unilaterally.

We next evaluate the convergence time of the
distributed computation offloading algorithm in Figure \ref{fig:Iterations}. It
shows that the average number of decision slots for convergence increases (almost) linearly as the
number of mobile device users $N$ increases. This demonstrates that the distributed computation offloading algorithm converges in a fast manner and scales well with the size of mobile device users in practice\footnote{For example, the length of a slot is at the time scale of microseconds in LTE system \cite{innovations2010lte} and hence the convergence time of the proposed algorithm is very short.}.

\section{Related Work}\label{sec:Related-Work}
Many previous work has investigated the single-user computation offloading problem (e.g., \cite{barbera2013offload,rudenko1998saving,huerta2008adaptable,xian2007adaptive,huang2012dynamic,wen2012energy,wu2013making}).  Barbera \emph{et al.} in \cite{barbera2013offload} showed by realistic measurements that the wireless access plays a key role in affecting the performance of mobile cloud computing. Rudenko \emph{et al.} in \cite{rudenko1998saving} demonstrated by experiments that significant energy can be saved by computation offloading.  Gonzalo \emph{et al.}  in \cite{huerta2008adaptable} developed an adaptive offloading algorithm based on both the execution history of applications and the current system conditions. Xian \emph{et al.}  in \cite{xian2007adaptive} introduced an efficient timeout scheme for computation offloading to increase the energy efficiency on mobile devices. Huang \emph{et al.} in \cite{huang2012dynamic} proposed a Lyapunov optimization based dynamic offloading algorithm to improve the mobile cloud computing performance while meeting the application execution time. Wen \emph{et al.} in \cite{wen2012energy} presented an efficient offloading policy by jointly configuring the clock frequency in the mobile device and scheduling the data transmission to minimize the energy consumption. Wu \emph{et al.} in \cite{wu2013making} applied the alternating renewal process to model the network availability and developed offloading decision algorithm accordingly.

To the best of our knowledge, only a few works have addressed the computation offloading problem under the setting of multiple mobile device users  \cite{barbarossa2013joint}.  Yang \emph{et al.} in \cite{yang2013framework} studied the scenario that multiple users share the wireless network bandwidth, and solved the problem of maximizing the mobile cloud computing performance by a centralized heuristic genetic algorithm. Our previous work in \cite{Chen2014DCO} considered the multi-user computation offloading problem in a single-channel wireless setting, such that each user has a binary decision variable (i.e., to offload or not). Given the fact that base-stations in most wireless networks are operating in the multi-channel wireless environment, in this paper we study the generalized multi-user computation offloading problem in a multi-channel setting, which results in significant differences in analysis. For example, we show the generalized problem is NP-hard, which is not true for the single-channel case. We also investigate the price of anarchy in terms of two performance metrics and show that the number of available channels can also impact the price of anarchy (e.g., Theorem \ref{thm:PoA2}). We further derive the upper bound of the convergence time of the computation offloading algorithm in the multi-channel environment.   Barbarossa \emph{et al.} in \cite{barbarossa2013joint} studied the multi-user computation offloading problem in a multi-channel wireless environment, by assuming that the number of wireless access channels is greater than the number of users such that each mobile user can offload the computation via a single  orthogonal channel independently without experiencing any interference from other users. In this paper we consider the more practical case that the number of wireless access channels is limited and each user mobile may experience interference from other users for computation offloading.


\section{Conclusion}\label{sec:Conclusion}
In this paper, we propose a game theoretic approach for the computation offloading decision making problem among multiple mobile device users for  mobile-edge cloud computing. We formulate the problem as as a  multi-user computation offloading game and show that the game always admits a Nash equilibrium. We also design a distributed computation offloading algorithm that can achieve a Nash equilibrium, derive the upper bound of convergence time, and quantify its price of anarchy. Numerical results demonstrate that the proposed algorithm achieves superior computation offloading performance and scales well as the user size increases.

For the future work, we are going to consider the more general case that mobile users may depart and leave dynamically within a computation offloading period. In this case, the user mobility patterns will play an important role in the problem formulation. Another direction is to study the joint power control and offloading decision making problem, which would be very interesting and technically challenging. 
\bibliographystyle{ieeetran}
\bibliography{MobileCloud}

\begin{thebibliography}{10}
\providecommand{\url}[1]{#1}
\csname url@samestyle\endcsname
\providecommand{\newblock}{\relax}
\providecommand{\bibinfo}[2]{#2}
\providecommand{\BIBentrySTDinterwordspacing}{\spaceskip=0pt\relax}
\providecommand{\BIBentryALTinterwordstretchfactor}{4}
\providecommand{\BIBentryALTinterwordspacing}{\spaceskip=\fontdimen2\font plus
\BIBentryALTinterwordstretchfactor\fontdimen3\font minus
  \fontdimen4\font\relax}
\providecommand{\BIBforeignlanguage}[2]{{%
\expandafter\ifx\csname l@#1\endcsname\relax
\typeout{** WARNING: IEEEtran.bst: No hyphenation pattern has been}%
\typeout{** loaded for the language `#1'. Using the pattern for}%
\typeout{** the default language instead.}%
\else
\language=\csname l@#1\endcsname
\fi
#2}}
\providecommand{\BIBdecl}{\relax}
\BIBdecl

\bibitem{kumar2010cloud}
K.~Kumar and Y.~Lu, ``Cloud computing for mobile users: Can offloading
  computation save energy?'' \emph{IEEE Computer}, vol.~43, no.~4, pp. 51--56,
  2010.

\bibitem{soyata2012cloud}
T.~Soyata, R.~Muraleedharan, C.~Funai, M.~Kwon, and W.~Heinzelman,
  ``Cloud-vision: Real-time face recognition using a mobile-cloudlet-cloud
  acceleration architecture,'' in \emph{IEEE ISCC}, 2012.

\bibitem{cohen2008embedded}
J.~Cohen, ``Embedded speech recognition applications in mobile phones: Status,
  trends, and challenges,'' in \emph{IEEE ICASSP}, 2008.

\bibitem{cuervo2010maui}
E.~Cuervo, A.~Balasubramanian, D.~Cho, A.~Wolman, S.~Saroiu, R.~Chandra, and
  P.~Bahl, ``{MAUI}: making smartphones last longer with code offload,'' in
  \emph{the 8th international conference on Mobile systems, applications, and
  services}, 2010.

\bibitem{satyanarayanan2009case}
M.~Satyanarayanan, P.~Bahl, R.~Caceres, and N.~Davies, ``The case for vm-based
  cloudlets in mobile computing,'' \emph{IEEE Pervasive Computing}, vol.~8,
  no.~4, pp. 14--23, 2009.

\bibitem{MEC2014}
{European Telecommunications Standards Institute}, ``Mobile-edge computing --
  introductory technical white paper,'' September 2014.

\bibitem{drolia2013case}
U.~Drolia, R.~Martins, J.~Tan, A.~Chheda, M.~Sanghavi, R.~Gandhi, and
  P.~Narasimhan, ``The case for mobile edge-clouds,'' in \emph{IEEE 10th
  International Conference on Ubiquitous Intelligence and Computing}.\hskip 1em
  plus 0.5em minus 0.4em\relax IEEE, 2013, pp. 209--215.

\bibitem{TelecomCloud2012}
\BIBentryALTinterwordspacing
Ericsson, ``The telecom cloud opportunity,'' March 2012. [Online]. Available:
  \url{http://www.ericsson.com/res/site_AU/docs/2012/ericsson_telecom_cloud_discussion_paper.pdf}
\BIBentrySTDinterwordspacing

\bibitem{barbarossa2013joint}
S.~Barbarossa, S.~Sardellitti, and P.~Di~Lorenzo, ``Joint allocation of
  computation and communication resources in multiuser mobile cloud
  computing,'' in \emph{IEEE Workshop on SPAWC}, 2013.

\bibitem{barbera2013offload}
M.~V. Barbera, S.~Kosta, A.~Mei, and J.~Stefa, ``To offload or not to offload?
  the bandwidth and energy costs of mobile cloud computing,'' in \emph{IEEE
  INFOCOM}, 2013.

\bibitem{rudenko1998saving}
A.~Rudenko, P.~Reiher, G.~J. Popek, and G.~H. Kuenning, ``Saving portable
  computer battery power through remote process execution,'' \emph{ACM
  SIGMOBILE Mobile Computing and Communications Review}, vol.~2, no.~1, pp.
  19--26, 1998.

\bibitem{huerta2008adaptable}
G.~Huertacanepa and D.~Lee, ``An adaptable application offloading scheme based
  on application behavior,'' in \emph{22nd International Conference on Advanced
  Information Networking and Applications-Workshops}, 2008.

\bibitem{xian2007adaptive}
C.~Xian, Y.~Lu, and Z.~Li, ``Adaptive computation offloading for energy
  conservation on battery-powered systems,'' in \emph{IEEE ICDCS},
  vol.~2.\hskip 1em plus 0.5em minus 0.4em\relax IEEE, 2007, pp. 1--8.

\bibitem{huang2012dynamic}
D.~Huang, P.~Wang, and D.~Niyato, ``A dynamic offloading algorithm for mobile
  computing,'' \emph{IEEE Transactions on Wireless Communications}, vol.~11,
  no.~6, pp. 1991--1995, 2012.

\bibitem{wen2012energy}
Y.~Wen, W.~Zhang, and H.~Luo, ``Energy-optimal mobile application execution:
  Taming resource-poor mobile devices with cloud clones,'' in \emph{IEEE
  INFOCOM}, 2012.

\bibitem{wu2013making}
H.~Wu, D.~Huang, and S.~Bouzefrane, ``Making offloading decisions resistant to
  network unavailability for mobile cloud collaboration,'' in \emph{IEEE
  Collaboratecom}, 2013.

\bibitem{Chen2014DCO}
X.~Chen, ``Decentralized computation offloading game for mobile cloud
  computing,'' \emph{IEEE Transactions on Parallel and Distributed Systems},
  2014.

\bibitem{wu2002multi}
S.~Wu, Y.~Tseng, C.~Lin, and J.~Sheu, ``A multi-channel mac protocol with power
  control for multi-hop mobile ad hoc networks,'' \emph{The Computer Journal},
  vol.~45, no.~1, pp. 101--110, 2002.

\bibitem{iosifidisiterative2013}
G.~Iosifidis, L.~Gao, J.~Huang, and L.~Tassiulas, ``An iterative double auction
  mechanism for mobile data offloading,'' in \emph{IEEE WiOpt}, 2013.

\bibitem{lopez2013distributed}
D.~L{\'o}pez-P{\'e}rez, X.~Chu, A.~V. Vasilakos, and H.~Claussen, ``On
  distributed and coordinated resource allocation for interference mitigation
  in self-organizing lte networks,'' \emph{IEEE/ACM Transactions on
  Networking}, vol.~21, no.~4, pp. 1145--1158, 2013.

\bibitem{rappaport1996wireless}
T.~S. Rappaport, \emph{Wireless communications: principles and practice}.\hskip
  1em plus 0.5em minus 0.4em\relax Prentice Hall PTR New Jersey, 1996.

\bibitem{xiao2003utility}
M.~Xiao, N.~B. Shroff, and E.~K. Chong, ``A utility-based power-control scheme
  in wireless cellular systems,'' \emph{IEEE/ACM Transactions on Networking},
  vol.~11, no.~2, pp. 210--221, 2003.

\bibitem{chiang2008power}
M.~Chiang, P.~Hande, T.~Lan, and C.~W. Tan, ``Power control in wireless
  cellular networks,'' \emph{Foundations and Trends in Networking}, vol.~2,
  no.~4, pp. 381--533, 2008.

\bibitem{yang2013framework}
L.~Yang, J.~Cao, Y.~Yuan, T.~Li, A.~Han, and A.~Chan, ``A framework for
  partitioning and execution of data stream applications in mobile cloud
  computing,'' \emph{ACM SIGMETRICS Performance Evaluation Review}, vol.~40,
  no.~4, pp. 23--32, 2013.

\bibitem{wallenius2008multiple}
J.~Wallenius, J.~S. Dyer, P.~C. Fishburn, R.~E. Steuer, S.~Zionts, and K.~Deb,
  ``Multiple criteria decision making, multiattribute utility theory: recent
  accomplishments and what lies ahead,'' \emph{Management Science}, vol.~54,
  no.~7, pp. 1336--1349, 2008.

\bibitem{hu2014quality}
W.~Hu and G.~Cao, ``Quality-aware traffic offloading in wireless networks,'' in
  \emph{ACM Mobihoc}, 2014.

\bibitem{loh2009solving}
K.-H. Loh, B.~Golden, and E.~Wasil, ``Solving the maximum cardinality bin
  packing problem with a weight annealing-based algorithm,'' in
  \emph{Operations Research and Cyber-Infrastructure}.\hskip 1em plus 0.5em
  minus 0.4em\relax Springer, 2009.

\bibitem{monderer1996potential}
D.~Monderer and L.~S. Shapley, ``Potential games,'' \emph{Games and economic
  behavior}, vol.~14, no.~1, pp. 124--143, 1996.

\bibitem{innovations2010lte}
T.~Innovations, ``{LTE} in a nutshell,'' \emph{White Paper}, 2010.

\bibitem{dey2013addressing}
S.~Dey, Y.~Liu, S.~Wang, and Y.~Lu, ``Addressing response time of cloud-based
  mobile applications,'' in \emph{Proceedings of the first international
  workshop on Mobile cloud computing and networking}, 2013.

\bibitem{roughgarden2005selfish}
T.~Roughgarden, \emph{Selfish routing and the price of anarchy}.\hskip 1em plus
  0.5em minus 0.4em\relax MIT press, 2005.

\bibitem{andrews2014will}
J.~G. Andrews, S.~Buzzi, W.~Choi, S.~V. Hanly, A.~Lozano, A.~C. Soong, and
  J.~C. Zhang, ``What will 5g be?'' \emph{IEEE Journal on Selected Areas in
  Communications}, vol.~32, no.~6, pp. 1065--1082, 2014.

\bibitem{bahl2009white}
P.~Bahl, R.~Chandra, T.~Moscibroda, R.~Murty, and M.~Welsh, ``White space
  networking with wi-fi like connectivity,'' \emph{ACM SIGCOMM Computer
  Communication Review}, vol.~39, no.~4, pp. 27--38, 2009.

\bibitem{quek2013small}
T.~Q. Quek, G.~de~la Roche, I.~G{\"u}ven{\c{c}}, and M.~Kountouris, \emph{Small
  cell networks: Deployment, PHY techniques, and resource management}.\hskip
  1em plus 0.5em minus 0.4em\relax Cambridge University Press, 2013.

\bibitem{rubinstein2004cross}
R.~Y. Rubinstein and D.~P. Kroese, \emph{The cross-entropy method: a unified
  approach to combinatorial optimization, Monte-Carlo simulation and machine
  learning}.\hskip 1em plus 0.5em minus 0.4em\relax Springer, 2004.

\end{thebibliography}

\end{document}